\documentclass[10pt]{article}
\usepackage{graphicx}
\usepackage{epsfig}
\usepackage{amssymb}
\usepackage{amsmath}
\usepackage{amsthm}
\newtheorem{lemma}{Lemma}[section]
\newtheorem{theorem}{Theorem}[section]

\newenvironment{remark}{\begin{description} \item[Remark.] }%
{\end{description}}
{\end{description}}

{\hfill $\bullet$ \\}

\newcommand{\be}{\begin{equation}}
\newcommand{\ee}{\end{equation}}

\begin{document}
    \title{An efficient explicit approach for predicting the Covid-19 spreading with undetected infectious: The case of Cameroon}
    \author{Eric Ngondiep}
   \date{$^{\text{\,1\,}}$\small{Department of Mathematics and Statistics, College of Science, Al-Imam Muhammad Ibn Saud\\ Islamic University
        (IMIU), $90950$ Riyadh $11632,$ Saudi Arabia.}\\
     \text{\,}\\
       $^{\text{\,2\,}}$\small{Hydrological Research Centre, Institute for Geological and Mining Research, 4110 Yaounde-Cameroon.}\\
     \text{\,}\\
        \textbf{Email addresses:} ericngondiep@gmail.com/engondiep@imamu.edu.sa}

    \maketitle
   \textbf{Abstract.}
   This paper considers an explicit numerical scheme for solving the mathematical model of the propagation of Covid-19 epidemic with undetected infectious cases. We analyze the stability and convergence rate of the new approach in $L^{\infty}$-norm. The proposed method is less time consuming. Furthermore, the method is stable, at least second-order convergent and can serve as a robust tool for the integration of general systems of ordinary differential equations. A wide set of numerical evidences which consider the case of Cameroon are presented and discussed.
   \text{\,} \\
    \text{\,}\\
   \ \noindent {\bf Keywords: mathematical model of SARS-Cov-2 spreading, a two-level explicit scheme, stability analysis, convergence rate,
   numerical experiments.} \\
   \\
   {\bf AMS Subject Classification (MSC). 65M10, 65M05}.

      \section{Introduction and motivation}\label{sec1}
      Deterministic models are important decision tools that can be useful to forecasting different scenarios. The first motivation of studying such
      models is based on the use of the theory of ordinary/partial differential equations and a low computational complexity which can permit a
      better calibration of the model characteristics. Furthermore, deterministic approaches are the only suitable methods that can be used when
      modeling a new problem with few data. For more details, we refer the readers to \cite{1ifvr,24ifvr,5ifvr,43ifvr} and references therein.
      The use of the mathematical models as a predictive tool in the simulation of complex problems arising in a broad range of practical applications
      in biology, environmental fluid mechanic, chemistry and applied mathematics (for example: mathematical model in population biology and
      epidemiology, mixed Stokes-Darcy model, Navier-Stokes equations, nonlinear time-dependent reaction-diffusion problem, heat conduction equation
      and unsteady convection-diffusion-reaction equations) represents a good candidate for developing efficient numerical schemes in the approximate
       solutions of such problems \cite{en11,en2,tg4,tg8,en10,en3,en4,tg10,gv1983,en6,en7,tg15,tg13,nkma,tg17,zl1984}. For parabolic partial
      differential equations (PDEs) which present strong steep gradients (for instance: shallow water flow and advection-diffusion equations), numerical
      algorithms are needed with good resolution of steep gradients \cite{14sn,3sn,11sn,13tls,25tls,30tls,18tls,30mt}.\\

      Early in an epidemic, the quality of the data on infections, deaths, tests and other factors often are limited by undetection or inconsistent
      detention of cases, reporting delays, and poor documentation, all of which affect the quality of any model output. Simpler models may provide
      less valid predictions since they cannot capture complex and unobserved human mixing patterns and other time-varying parameters of infectious
      disease spread. Also, complex models may be no more reliable than simpler ones if they miss key aspects of the biological entities (either ions,
      molecules, proteins or ceils) \cite{3jlj}. At a time when numbers of cases and deaths from coronavirus $2019$ (Covid-19) pandemic continue to
      increase with alarming speed, accurate forecasts from mathematical models are increasingly important for physicians, politicians, epidemiologists,
      the public and most importantly, for authorities responsible of organizing care for the populations they serve. Given the unpredictable behavior
      of severe acute respiratory syndrome Covid-19, it is worth mentioning that efficient numerical approaches are the best tools that can be used to
      predict the spread of the disease with reasonable accuracy. These predictions have crucial consequences regarding how quickly and strongly the
      government of a country mores to curb a pandemic. However, assuming the worst-case scenario at state and national levels will lead to
      inefficiencies (such as: the competition for beds and supplies) and may compromise effective delivery and quality of care, whereas supposing
      the best-case scenario can conduct to disastrous underpreparation.\\

      Covid-19 is a rapidly spreading infectious disease caused by the novel coronavirus SARS-Cov-2, a betacoronavirus which has provided a global
      epidemic. Up today, no drug to treat the Covid-19 disease is officially available (approved by the World Health Organization (WHO)) and a vaccine
      will not be available for several months at the earliest. The only approaches widely used to slow the spread of the pandemic are those of
      classical epidemic control such as: physical distancing, contact tracing, hygiene measure, quarantine and case of isolation. However, the primary
      and most effective use of the epidemiologic models is to estimate the relative effect of various interventions in reducing disease burden rather
      than to produce precise quantitative predictions about extent on duration of disease burdens. Nevertheless, consumers of such models including
      the media, the publics and politicians sometimes focus on the quantitative predictions of infectious and mortality estimates. Such measures of
       potential disease burden are also necessary for planners who consider future outcomes in light of health care capacity. The big challenge consists
      to assess such estimates.\\

      In this paper, we develop an efficient numerical scheme for solving a mathematical model well adapted to Covid-19 pandemic subjected to special
      characteristics (effect of undetected infected cases, effect of different sanitary and infectiousness conditions of hospitalized people and
      estimation of the needs of beds in hospitals) and considering different scenarios \cite{ifvr}. Specifically, the proposed technique should
      provide the numbers of detected infected and undetected infected cases, numbers of deaths and needs of beds in hospitals in countries (for
      example, in Cameroon) where Covid-19 is a very serious health problem. It is Worth noticing that the model of Covid-19 considered in this work
      has been obtained under the asumption of "only within-country disease spread" for territories with relevant number of people infected by
      SARS-Cov-2, where local transmission is the major cause of the disease spread (for instance: case of Cameroon). Furthermore, the parameters
      of the model used in this note are taken from the literature \cite{20ifvr,21ifvr,ifvr}. Our study also relates the disease fatality rate with
      the percentage of detected cases over the real total infected cases which allows to analyze  the importance of this percentage on the impact
      of Covid-19. In addition, to demonstrate the efficiency and validity of the new approach when applied to the mathematical problem of coronavirus
      $2019$ epidemic, we consider the case of Cameroon, the country of the central Africa where one can observe the highest number of people infected
      by the new virus SARS-Cov-2. We compare the results produced by the numerical method to the data obtained from this country and those provided by
      the World Health Organization in its reports \cite{31ifvr}. Finally, it important to mention that the considered area (Cameroon) in the
      numerical experiments can be replaced by any territory worldwide.\\

      This paper is organized as follows: Section $\ref{sec2}$ considers some preliminaries together with the mathematical formulation of Covid-19
      spreading. In section $\ref{sec3}$, we provide a full description of the two-level explicit scheme for solving the problem indicated
      in section $\ref{sec2}$. Section $\ref{sec4}$ analyzes the stability and the convergence rate of the new procedure while a large set of
      numerical experiments are presented and critically discussed in Section $\ref{sec5}$. We draw in section $\ref{sec6}$ the general conclusion and
      we provided our future investigations.

      \section{Preliminaries and mathematical model of SARS-Cov-2 infectiousness}\label{sec2}
      We use a mathematical formalism \cite{ifvr} that describes how infectiousness varies as a function of time since infections for a representative
      cohort of infected persons. We assume that transmission of SARS-Cov-2 is contagious from person to person and not point source. Furthermore, it is
      also assumed that, at the initial phase of Covid-19 disease, the proportion of the population with immunity to SARS is negligible
      \cite{13cp,14cp,15cp,10cp}. At the beginning of a contagious epidemic, a small number of infected people start passing the disease to a large
       population. Individuals can go through nine states. They start out susceptible ($X_{1}$: the person is not infected by the disease pathogen),
       exposed ($X_{2}$: the person is in the incubation period after being infected by the disease pathogen, but has no clinical signs), infected
       ($X_{3}$: the person has finished the incubation period, may infected other people and start developing the clinical signs. Here, people can
       be taken in charge by sanitary authorities of this country (hospitalized persons) or not detected by the authorities and continue as infectious),
       infectious but undetected ($X_{4}$: the person can still infect other individuals, have clinical signs, but is not detected and reported by the
       authorities (these people will not die)), hospitalized or in quarantine at home ($X_{5}$: the person is in hospital or in quarantine at home,
       can still infect other people, but will recover), hospitalized but will die ($X_{6}$: the person is hospitalized and can infect other people,
       but will die), recovered after being previously infectious but undetected ($X_{8}$: the person was not previously detected as infectious,
       survived the disease, is no longer infectious and has developed a natural immunity to the disease), recovered after being previously detected
       as infectious ($X_{7}$: the person survived the disease, is no longer infectious and has developed a natural immunity to the virus, but she/he
        remains in hospital for a convalescence period $d_{0}$ days), dead by SARS-Cov-2 ($X_{9}$).\\

      The proposed model is based on thirteen parameters.

      \begin{itemize}
        \item $R_{0}$ denotes basis reproductive number, that is, the expected number of new infectious cases per infectious case,
        \item $N$ is the number of persons in a considered country before the starting of the pandemic,
        \item $\mu_{1}\in[0,1],$ designates the natality rate ($day^{-1}$) in the considered country (the number of births per day and per capita),
        \item $\mu_{2}\in[0,1],$ represents the mortality rate ($day^{-1}$) in the considered country (the number of deaths per day and per capita),
        \item $w(t)\in[\underline{w},\overline{w}]\subset[0,1],$ denotes the case fatality rate in the considered territory at time $t$ (the proportion of deaths compared to the total number of infectious people (detected or undetected). Here, $\underline{w}$ and $\overline{w}$ are the minimum and maximun case fatality rates in the country, respectively),
        \item $\theta(t)\in[\overline{w},1],$ means the fraction of infected people that are detected and reported by the authorities in the country at time $t.$ For the convenience of writting, we assume that all the deaths due to Covid-19 are detected and reported, so $\theta(t)\geq\overline{w}$,
        \item $\beta_{X_{j}}\in\mathbb{R}^{+},$ for $j=2,3,\cdots,6,$ are the disease contact rates ($day^{-1}$) of a person in the corresponding compartment $X_{j}$, in the country (without taking into account the control measures),
        \item $\beta_{X_{4}}(\theta)\in\mathbb{R}^{+},$ represents the disease contact rates ($day^{-1}$) of a person in compartment $X_{4}$, in the country (without taking into account the control measures), where the fraction of infected individuals that are detected is $\theta(t)$,
        \item $\gamma_{X_{2}}\in(0,\infty),$ designates the transition rate ($day^{-1}$) from compartment $X_{2}$ to compartment $X_{3}$. It's the same in all the countries,
        \item $\gamma_{X_{3}}(t)\in(0,\infty),$ is the transition rate ($day^{-1}$) from compartment $X_{3}$ to compartments $X_{4}$, $X_{5}$ or $X_{6}$ at time $t$. It can change from a country to another,
        \item $\gamma_{X_{4}}(t)$, $\gamma_{X_{5}}(t)$ and $\gamma_{X_{6}}(t)\in(0,\infty),$ denote the transition rate ($day^{-1}$) from compartments $X_{4}$, $X_{5}$ or $X_{6}$ to compartments $X_{7}$, $X_{8}$ and $X_{9}$, respectively, in the considered country at time $t$,
        \item $m_{X_{j}}(t)\in[0,1],$ for $j=2,3,\cdots,6,$ are functions representing the efficiency of the control measures applied to the corresponding compartment in the considered country at time $t,$
        \item $\tau_{1}$ is the person infected that arrives in the territory from other countries per day. $\tau_{2}$ is the person infected that leaves the territory from other countries per day. Both can be modeled following the between-country spread part of the Be-CoDis model, see \cite{16ifvr}.
      \end{itemize}
      The control measures applied by the government to curb the Covid-19 spread are those provided by the WHO in \cite{4ifvr,35ifvr}:
       \begin{itemize}
        \item isolation: infected people are isolated from contact with other persons. Only sanitary professionals are in contact with them. Isolated patients receive an adequate medical treatment that reduces the Covid-19 fatality rate,
        \item quarantine: movement of people in the area of origin of an infected person is restricted and controlled (for instance: quick sanitary check-points at the airports) to avoid that possible infected people spread the disease,
        \item tracing: the aim of tracing is to identify potential infectious contacts which may have infected an individual or spread SARS-Cov-2 to other people. Increase the number of tests in order to increase the percentage of detected infected persons,
        \item increase of sanitary resources: number of operational beds and sanitary personal available to detect and treat affected people is increased, producing a decrease in the infectious period for the compartment $X_{3}.$
      \end{itemize}
      Furthermore, the mathematical model of coronavirus $2019$ epidemic considers the following assumptions:
      \begin{description}
        \item[$(a_{1})$] the population at risk is large enough and time period of concern is short enough that over the time period of interest, very close to $100\%$ of the population is susceptible,
        \item[$(a_{2})$] the pandemic is at the early stage and has not reached the point where the susceptible population decreases so much due to death or post-infection immunity that the average number of secondary cases falls,
        \item[$(a_{3})$] unprotected contact results in infection,
        \item[$(a_{4})$] the epidemic in the population of interest begins with a single host (note that the equations used in computing cases and deaths are easily modified if this is not the case),
        \item[$(a_{5})$] infectivity occurs during the incubation period only,
        \item[$(a_{6})$] the models are deterministic, that is, the thirteen parameters of Covid-19 spread cited above are constant values.
      \end{description}
      Under these assumptions, the mathematical formulation of Covid-19 disease is given by the following system of nonlinear ordinary differential
        equations:
      \begin{equation*}
       \frac{dX_{1}}{dt}=-\frac{X_{1}}{N}\left[m_{X_{2}}(t)\beta_{X_{2}}(t)X_{2}+m_{X_{3}}(t)\beta_{X_{3}}(t)X_{3}+m_{X_{4}}(t)\beta_{X_{4}}(\theta)X_{2}+
       m_{X_{5}}(t)\beta_{X_{5}}(t)X_{5}+m_{X_{6}}(t)\beta_{X_{6}}(t)X_{6}\right]
      \end{equation*}
      \begin{equation}\label{1}
        -\mu_{m}X_{1}+\mu_{n}[X_{1}+X_{2}+X_{3}+X_{4}+X_{7}+X_{8}],
      \end{equation}
       \begin{equation*}
       \frac{dX_{2}}{dt}=\frac{X_{1}}{N}\left[m_{X_{2}}(t)\beta_{X_{2}}(t)X_{2}+m_{X_{3}}(t)\beta_{X_{3}}(t)X_{3}+m_{X_{4}}(t)\beta_{X_{4}}(\theta)X_{2}+
       m_{X_{5}}(t)\beta_{X_{5}}(t)X_{5}+m_{X_{6}}(t)\beta_{X_{6}}(t)X_{6}\right]
      \end{equation*}
      \begin{equation}\label{2}
        -\mu_{m}X_{2}-\gamma_{X_{2}}(t)X_{2}+\tau_{1}(t)-\tau_{2}(t),
      \end{equation}
       \begin{equation}\label{3}
       \frac{dX_{3}}{dt}=\gamma_{X_{2}}(t)X_{2}-\left(\mu_{m}+\gamma_{X_{3}}(t)\right)X_{3},\text{\,\,\,\,}
       \frac{dX_{4}}{dt}=(1-\theta(t))\gamma_{X_{3}}(t)X_{3}-\left(\mu_{m}+\gamma_{X_{4}}(t)\right)X_{4},
      \end{equation}
      \begin{equation}\label{4}
       \frac{dX_{5}}{dt}=(\theta(t)-w(t))\gamma_{X_{3}}(t)X_{3}-\gamma_{X_{5}}(t)X_{5},\text{\,\,\,\,}
       \frac{dX_{6}}{dt}=w(t)\gamma_{X_{3}}(t)X_{3}-\gamma_{X_{6}}(t)X_{6},
      \end{equation}
      \begin{equation}\label{5}
       \frac{dX_{7}}{dt}=\gamma_{X_{5}}(t)X_{5}-\mu_{m}X_{7},\text{\,\,\,\,}
       \frac{dX_{8}}{dt}=\gamma_{X_{4}}(t)X_{4}-\mu_{m}X_{8}\text{\,\,\,\,and\,\,\,\,}\frac{dX_{9}}{dt}=\gamma_{X_{6}}(t)X_{6},
      \end{equation}
      with the initial conditions
      \begin{equation}\label{6}
      X_{j}(t_{0})=X_{j}^{0}\in(0,\infty),\text{\,\,\,for\,\,\,}j=1,2,\cdots,9,
      \end{equation}
      where all the unknowns $X_{j}$ dependent on the time $t\in[t_{0},T_{max}].$ Setting $X(t)=(X_{1},X_{2},\ldots,X_{9})^{T}$ and
      $F(t,X(t))=(F_{1}(t,X(t)),F_{2}(t,X(t)),\ldots,F_{9}(t,X(t)))^{T},$ where
      \begin{equation*}
       F_{1}(t,X(t))=-\frac{X_{1}}{N}\left[m_{X_{2}}(t)\beta_{X_{2}}(t)X_{2}+m_{X_{3}}(t)\beta_{X_{3}}(t)X_{3}+m_{X_{4}}(t)\beta_{X_{4}}(\theta)X_{2}+
       m_{X_{5}}(t)\beta_{X_{5}}(t)X_{5}\right.
      \end{equation*}
      \begin{equation}\label{10}
        +\left.m_{X_{6}}(t)\beta_{X_{6}}(t)X_{6}\right]-\mu_{m}X_{1}+\mu_{n}[X_{1}+X_{2}+X_{3}+X_{4}+X_{7}+X_{8}],
      \end{equation}
       \begin{equation*}
       F_{2}(t,X(t))=\frac{X_{1}}{N}\left[m_{X_{2}}(t)\beta_{X_{2}}(t)X_{2}+m_{X_{3}}(t)\beta_{X_{3}}(t)X_{3}+m_{X_{4}}(t)\beta_{X_{4}}(\theta)X_{2}+
       m_{X_{5}}(t)\beta_{X_{5}}(t)X_{5}\right.
      \end{equation*}
      \begin{equation}\label{11}
        +\left.m_{X_{6}}(t)\beta_{X_{6}}(t)X_{6}\right]-\mu_{m}X_{2}-\gamma_{X_{2}}(t)X_{2}+\tau_{1}(t)-\tau_{2}(t),
      \end{equation}
       \begin{equation}\label{12}
       F_{3}(t,X(t))=\gamma_{X_{2}}(t)X_{2}-\left(\mu_{m}+\gamma_{X_{3}}(t)\right)X_{3},\text{\,\,\,}
       F_{4}(t,X(t))=(1-\theta(t))\gamma_{X_{3}}(t)X_{3}-\left(\mu_{m}+\gamma_{X_{4}}(t)\right)X_{4},
      \end{equation}
      \begin{equation}\label{13}
       F_{5}(t,X(t))=(\theta(t)-w(t))\gamma_{X_{3}}(t)X_{3}-\gamma_{X_{5}}(t)X_{5},\text{\,\,\,}
       F_{6}(t,X(t))=w(t)\gamma_{X_{3}}(t)X_{3}-\gamma_{X_{6}}(t)X_{6},
      \end{equation}
      \begin{equation}\label{14}
       F_{7}(t,X(t))=\gamma_{X_{5}}(t)X_{5}-\mu_{m}X_{7},\text{\,\,\,}F_{8}(t,X(t))=\gamma_{X_{4}}(t)X_{4}-\mu_{m}X_{8}\text{\,\,\,and\,\,\,}
      F_{9}(t,X(t))=\gamma_{X_{6}}(t)X_{6},
      \end{equation}
      the system of nonlinear equations $(\ref{1})$-$(\ref{5})$ is equivalent to
      \begin{equation}\label{s1}
      \frac{dX}{dt}=F(t,X).
      \end{equation}
      \begin{remark}
      In the modeling point of view, the term $\frac{w(t)}{\theta(t)}$ corresponds to the apparent fatality rate of the disease (obtained by
      considering only the detected cases) in the considered area at time $t,$ whereas $w(t)$ is the real fatality rate of coronavirus $2019$ disease.
      \end{remark}
      Since the mathematical model of Covid-19 provided by the system of equations $(\ref{1})$-$(\ref{5})$ is too complex and because both natality
       and mortality (not from SARS-Cov-2) do not seem to be useful factors for this pandemic (at least for relatively short periods of time), we
       assume in the rest of this paper that
      \begin{equation}\label{19a}
      \mu_{m}=\mu_{n}=0.
      \end{equation}
      It is worth mentioning that the aim of this paper is to compute the following Covid-19 characteristics:
      \begin{description}
        \item[1)] the model cumulative of coronavirus $2019$ cases at day $t$ given by
        \begin{equation}\label{19b}
      c_{m}(t)=X_{5}(t)+X_{6}(t)+X_{8}(t)+X_{9}(t),
      \end{equation}
        \item[2)] the model cumulative number of deaths (due to Covid-19) at day $t$, which is given by $X_{9}(t)$,
        \item[3)] $R_{e}(t)$ which is the effective reproductive number of Covid-19,
        \item[4)] the number of people in hospital is estimated by the following equation
        \begin{equation}\label{20}
      Host(t)=X_{6}(t)+p(t)[X_{5}(t)+(X_{7}(t)-X_{7}(t-d_{0}))],
      \end{equation}
      where $p(t)$ represents the fraction, at time $t,$ of people in compartment $X_{5}$ that are hospitalized and $d_{0}$ days is the period
      of convalescence (i.e., the time a person is still hospitalized after recovering from Covid-19). This function can help to estimate and plan the
       number of clinical beds needed to treat all the SARS-Cov-2 cases at time $t$,
        \item[5)] the maximum number of hospitalized persons at the same time in the territory during the time interval $[t_{0},T_{max}],$ which is defined as
          \begin{equation}\label{21}
      MaxHost=\underset{t_{0}\leq t\leq T_{max}}{\max}Host(t).
      \end{equation}
      $MaxHost$ can help to estimate and plan the number of clinical beds needed to treat all the coronavirus $2019$ cases over the interval
       $[t_{0},T_{max}],$
        \item[6)] the number of people infected during the time interval $[t_{0},T_{max}],$ by contact with people in compartments $X_{2}$, $X_{4}$ and $X_{10}=X_{5}+X_{6}$, respectively. They are defined as
      \begin{equation}\label{22}
      \Gamma_{X_{2}}(t)=\frac{1}{N}\int_{t_{0}}^{T_{max}}m_{X_{2}}(s)\beta_{X_{2}}(s)X_{2}(s)X_{1}(s)ds,
      \end{equation}
      \begin{equation}\label{23}
      \Gamma_{X_{4}}(t)=\frac{1}{N}\int_{t_{0}}^{T_{max}}m_{X_{4}}(s)\beta_{X_{4}}(s)X_{4}(s)X_{1}(s)ds,
      \end{equation}
      \begin{equation}\label{24}
      \Gamma_{X_{10}}(t)=\frac{1}{N}\int_{t_{0}}^{T_{max}}\left(m_{X_{5}}(s)\beta_{X_{5}}(s)X_{5}(s)+m_{X_{6}}(s)\beta_{X_{6}}(s)X_{6}(s)\right)X_{1}(s)ds.
      \end{equation}
      \end{description}
      We recall that the basis reproduction number $R_{0}$  is defined as the number of cases an infected individual generates on average over the
      course of its infectious period, in an otherwise uninfected population and without special control measures. It depends on the considered
      population, but does change during the spread of the disease, while the effective reproduction number $R_{e}(t)$ is defined as the number of
      cases one infected person generates on average over the course of its infectious period. A part of the population can be already infected
      and/or special control measures that have been implemented. It depends on the spread of the disease. In addition, $R_{e}(t_{0})=R_{0},$ and
      the evolution of the epidemic slow down when $R_{e}(t)<1.$\\

      Now, applying the next generation method \cite{44ifvr} to the nonlinear system $(\ref{1})$-$(\ref{5})$ to get
      \begin{equation*}
      R_{0}=\left\{\gamma_{X_{6}}[((\beta_{X_{4}}(1-\theta)\gamma_{X_{5}}+\beta_{X_{5}}\gamma_{X_{4}}(\theta-w))\gamma_{X_{3}}+\beta_{X_{3}}\gamma_{X_{4}}
      \gamma_{X_{5}})\gamma_{X_{2}}+\beta_{X_{2}}\gamma_{X_{3}}\gamma_{X_{4}}\gamma_{X_{5}}]\right.
      \end{equation*}
      \begin{equation}\label{25}
      \left.+w\beta_{X_{6}}\gamma_{X_{2}}\gamma_{X_{3}}\gamma_{X_{4}}\gamma_{X_{5}}\right\}(\gamma_{X_{2}}\gamma_{X_{3}}\gamma_{X_{4}}
      \gamma_{X_{5}}\gamma_{X_{6}})^{-1},
      \end{equation}
       and
       \begin{equation*}
      R_{e}(t)=\frac{X_{1}(t)}{N}\left\{\gamma_{X_{6}}[((m_{X_{4}}\beta_{X_{4}}(1-\theta)\gamma_{X_{5}}+m_{X_{5}}\beta_{X_{5}}\gamma_{X_{4}}(\theta-w))\gamma_{X_{3}}
      +m_{X_{3}}\beta_{X_{3}}\gamma_{X_{4}}\gamma_{X_{5}})\gamma_{X_{2}}\right.
      \end{equation*}
      \begin{equation}\label{26}
      \left.+m_{X_{2}}\beta_{X_{2}}\gamma_{X_{3}}\gamma_{X_{4}}\gamma_{X_{5}}]+wm_{X_{6}}\beta_{X_{6}}\gamma_{X_{2}}\gamma_{X_{3}}\gamma_{X_{4}}\gamma_{X_{5}}\right\}(\gamma_{X_{2}}\gamma_{X_{3}}\gamma_{X_{4}}
      \gamma_{X_{5}}\gamma_{X_{6}})^{-1},
      \end{equation}
      where for the sake of simplicity of notations, all previous coefficients correspond to their particular values at times $t_{0}$ and $t,$
       respectively.\\

      In the literature \cite{23ifvr,48ifvr}, it is established that the observed patterns of Covid-19 are not completely consistent with the hypothesis
       that high absolute humidity may limit the survival and transmission of the virus, whereas the lower is the temperature, the greater is the
      survival period of the SARS-Cov-2 outside the host. Since there is no scientific evidence of the effect of the humidity and the temperature
       on the SARS-Cov-2, these factors are not included in our model.\\

      Focusing on the application on the control strategies, the efficiency of these measures indicated in \cite{19ifvr} satisfies equations
      \begin{equation}\label{28}
      m_{X_{j}}(t):=m(t)=\left\{
                                                   \begin{array}{ll}
                                                     (m_{l}-m_{l+1})\exp[-k_{l+1}(t-\lambda_{l})], & \hbox{$t\in[t_{l},\lambda_{l+1})$,
                                                     $l=0,1,\cdots,q-1$,}\\
                                                     \text{\,}\\
                                                     (m_{q-1}-m_{q})\exp[-k_{q}(t-\lambda_{q-1})], & \hbox{$t\in[t_{q-1},\infty)$,}
                                                   \end{array}
                                                 \right.
      \end{equation}
      for $j=1,2,\cdots,6$, where $m_{l}\in[0,1]$ measures the intensity of the control measures (greater value implies lower value of disease contact
       rates), $k_{l}\in[0,\infty)$ (in $day^{-1}$) simulates the efficiency of the control strategies (greater value implies lower value of disease
       contact rates) and $\lambda_{l}\in[t_{0},\infty)$, $l=1,2,\ldots,q$, denotes the first day of application of each control strategy.
        $\lambda_{0}\in[t_{0},\infty)$ is the first day of application of a control measure that was being used before $t_{0},$ if any. In this
        work, $q\in\mathbb{N}$ represents the number of changes of control strategy. In general, the values of $\lambda_{l}$ are typically taken
       in the literature (using dates when the countries implement special control measures). It is important to remind that some of the values
        of $m_{l}$ can be also sometimes known. The rest of the parameters needed to be calibrated.\\

      In the following, we assume that the case fatality rate $w(t)$, depends on the considered country, time $t$, and it can be affected by the
      application of the control measures (such as, earlier detection, better sanitary condition, etc...)\cite{16ifvr}. Thus, it satisfies equation
      \begin{equation}\label{29a}
      w(t)=m(t)\underline{w}+(1-m(t))\overline{w},
      \end{equation}
      where $\overline{w}\in[0,1]$ is the case fatality rate when no control measures are applied (i.e., $m(t)=1$) and $\underline{w}\in[0,1]$ is
      the case fatality rate when implemented control measures are fully applied ($m(t)=m_{q}$).\\

      Denoting by $d_{X_{3}}$, $d_{X_{4}}$, $d_{X_{5}}$ and $d_{X_{6}}$, be the "average" duration in days of a person in compartment $X_{3}$, $X_{4}$,
       $X_{5}$ and $X_{6}$, respectively, without the application of control strategies, we assume as in \cite{20ifvr,21ifvr} that
      \begin{itemize}
        \item the transition rate from $X_{2}$ to $X_{3}$ depends on the disease and, therefore is considered constant, that is $\gamma_{2}(t)=\alpha=c^{t},$
        \item the value of $\gamma_{X_{3}}(t):=\gamma(t)$ can be increased due to the application of control measures (that is, people with symptoms are detected earlier). As a consequence, the values of $\gamma_{X_{4}}(t)$,  $\gamma_{X_{5}}(t)$ and  $\gamma_{X_{6}}(t)$ can be decreased (i.e., persons with symptoms stay under observation during more time),
        \item $d_{X_{6}}=d_{X_{5}}+\delta,$    $\delta>0,$
        \item for the sake of readability, the infectious period for undetected individuals is the same than that of hospitalized people that survive the disease. So $d_{X_{4}}=d_{X_{5}}$. Furthermore, we suppose that the additional time a person is in the compartments $X_{3}$ and $X_{4}$ is constant, so it comes from \cite{16ifvr} that
            \begin{equation}\label{29}
            \gamma(t):=\gamma_{X_{3}}(t)=\frac{1}{d_{X_{3}}-g(t)},\text{\,\,\,\,\,\,}(day^{-1})
            \end{equation}
            \begin{equation}\label{30}
            \rho(t):=\gamma_{X_{4}}(t)=\gamma_{X_{5}}(t)=\frac{1}{d_{X_{4}}+g(t)},\text{\,\,\,\,\,\,}(day^{-1})
            \end{equation}
            \begin{equation}\label{31}
            \psi(t):=\gamma_{X_{6}}(t)=\frac{1}{d_{X_{4}}+g(t)+\delta},\text{\,\,\,\,\,\,}(day^{-1})
            \end{equation}
            where $g(t)=d_{g}(1-m(t))$ represents the decrease of the duration of the function $d_{X_{3}}$ due to the application of the
           control measures at time $t,$ $d_{g}$ is the maximum number of days that $d_{X_{3}}$ can be decreased due to the control measures.\\
             \end{itemize}
          Finally, the disease contact rate $\beta_{X_{4}}(\theta)$ is defined by
           \begin{equation}\label{32}
        \beta_{X_{4}}(\theta)=\left\{
                                \begin{array}{ll}
                                  \overline{\beta}_{X_{3}}, & \hbox{if $\theta=\overline{w}$,} \\
                                  \text{\,}\\
                                  nonincrease, & \hbox{} \\
                                  \text{\,}\\
                                  \underline{\beta}_{X_{3}}, & \hbox{if $\theta=\underline{w}$,}
                                \end{array}
                              \right.
       \end{equation}
       where $\underline{\beta}_{X_{3}}$ and $\overline{\beta}_{X_{3}}$ are suitable lower and upper bounds, respectively. For the convenience of
       writing, we assume that $\overline{\beta}_{X_{3}}=\beta_{X_{3}}:=\beta.$ In addition, the people in compartments $X_{2}$, $X_{4}$, $X_{5}$
       and $X_{6}$ are less infectious than people in compartment $X_{3}$ (due to their lower virus load or isolation measures). This fact results in
      \begin{equation}\label{33}
       \beta_{X_{2}}=c_{X_{2}}\beta_{X_{3}},\text{\,\,}\beta_{X_{5}}=\beta_{X_{6}}=c_{X_{10}}(t)\beta_{X_{3}},\text{\,\,}
       \underline{\beta}_{X_{3}}=c_{u}\beta_{X_{3}},
       \end{equation}
       where $c_{X_{2}},c_{X_{10}},c_{u}\in[0,1].$

       \section{Construction of the two-level explicit numerical scheme}\label{sec3}
       In this section, we develop the robust two-level explicit scheme for solving the mathematical problem $(\ref{1})$-$(\ref{5})$ modeling the spread
       of Covid-19 with undetected cases.\\

       Let $h:=\Delta t=\frac{T_{max}-t_{0}}{M}$ be the step size, $M$ is a positive integer. Set $t_{n}=t_{0}+nh,$
        $t_{n-\frac{1}{2}}=\frac{t_{n}+t_{n-1}}{2}$ for $n=1,2,...,M$ and $\Omega_{h}=\{t_{n},0\leq n\leq M\}$ be a regular partition of
       $[t_{0},T_{max}]$. Let $\mathcal{F}_{h}=\{X^{n}_{i},n=0,1,..., M;1\leq i\leq9\},$ be the grid functions space defined on
       $\Omega_{h}\times\mathbb{R}^{9}\subset I\times\mathbb{R}^{9}:=[t_{0},T_{max}]\times\mathbb{R}^{9}.$\\

       Define the following norms
       \begin{equation}\label{34}
        \|X^{n}\|_{\infty}=\underset{1\leq i\leq9}{\max}|X_{i}^{n}|\text{\,\,\,and\,\,\,}\||X|\|_{L^{2}(I)}=
        \left(h\underset{n=0}{\overset{M}\sum}\|X^{n}\|_{\infty}^{2}\right)^{\frac{1}{2}}
       \end{equation}
       where $|\cdot|$ designates the norm defined on the field of complex numbers $\mathbb{C}$. Furthermore, denote
       \begin{equation}\label{35}
        P_{j}^{(i)}(t,X(t))=\underset{l}{\sum}F_{i}(t_{l},X(t_{l}))L_{l}(t),
       \end{equation}
       where the function $L_{l}(t)$ is given by
       \begin{equation}\label{36}
        L_{l}(t)=\underset{\overset{q}{q\neq l}}{\prod}\frac{t-t_{q}}{t_{l}-t_{q}},
       \end{equation}
       be a polynomial of degree $j$ interpolating the function $F_{i}(t,X(t))$ at the node points $(t_{l},F_{i}(t_{l},X(t_{l})))$. According to
       equations $(\ref{35})$-$(\ref{36})$, it's important to remind that $P_{j}^{(i)}(t,X(t))$ is not necessarily the interpolation polynomial of
       degree $j$ of the function $F_{i}(t,X(t))$ at the node points $(t_{l},F_{i}(t_{l},X(t_{l})))$.\\

       Now, integrating both sides of equation $(\ref{s1})$ at the node points $t_{n}$ and $t_{n+\frac{1}{2}}$, this yields
       \begin{equation*}
        X(t_{n+\frac{1}{2}})-X(t_{n})=\int_{t_{n}}^{t_{n+\frac{1}{2}}}F(t,X)dt,
       \end{equation*}
       which is equivalent to
       \begin{equation}\label{37}
        X(t_{n+\frac{1}{2}})=X(t_{n})+\int_{t_{n}}^{t_{n+\frac{1}{2}}}F(t,X)dt.
       \end{equation}
       For $j=1,$ $P_{1}^{(i)}(t,X(t))$ is a linear polynomial approximating the function $F_{i}(t,X(t))$ at the points $(t_{n},F_{i}(t_{n},X(t_{n})))$ and
       $(t_{n+\frac{1}{2}},F_{i}(t_{n+\frac{1}{2}},X(t_{n+\frac{1}{2}})))$. Using equations $(\ref{35})$ and $(\ref{36})$, it is easy to observe that
       \begin{equation*}
        P_{1}^{(i)}(t,X(t))=F_{i}(t_{n},X(t_{n}))\frac{t-t_{n+\frac{1}{2}}}{t_{n}-t_{n+\frac{1}{2}}}+F_{i}\left(t_{n+\frac{1}{2}},X(t_{n+\frac{1}{2}})\right)
        \frac{t-t_{n}}{t_{n+\frac{1}{2}}-t_{n}}=\frac{2}{h}\left[\left(F_{i}\left(t_{n+\frac{1}{2}},X(t_{n+\frac{1}{2}})\right)-\right.\right.
       \end{equation*}
       \begin{equation}\label{38}
        \left.\left.F_{i}(t_{n},X(t_{n}))\right)t+t_{n+\frac{1}{2}}F_{i}(t_{n},X(t_{n}))-t_{n}F_{i}\left(t_{n+\frac{1}{2}},X(t_{n+\frac{1}{2}})
        \right)\right],
       \end{equation}
       where the error term is given by
       \begin{equation}\label{39}
        F_{i}(t,X(t))-P_{1}^{(i)}(t,X(t))=\frac{1}{2}(t-t_{n})(t-t_{n+\frac{1}{2}})\frac{d^{2}F_{i}}{dt^{2}}(t_{\epsilon},X(t_{\epsilon})):=O(h^{2}),
       \end{equation}
       where $t_{\epsilon}$ (respectively, each component $X_{i}(t_{\epsilon})$ of $X(t_{\epsilon})$) is an unknown function which is between the
       maximum and the minimum of the numbers $t_{n}$, $t_{n+\frac{1}{2}}$ and $t$ (respectively: $X_{i}(t_{n})$, $X_{i}(t_{n+\frac{1}{2}})$
        and $X_{i}(t)$). But equation $(\ref{39})$ can be rewritten as
       \begin{equation}\label{40}
        F_{i}(t,X(t))=P_{1}^{(i)}(t,X(t))+O(h^{2}),\text{\,\,\,for\,\,\,}i=1,2,...,9.
       \end{equation}
       Substituting approximation $(\ref{40})$ into the ith equation of the system $(\ref{37})$, we obtain
       \begin{equation}\label{41}
        X_{i}(t_{n+\frac{1}{2}})=X_{i}(t_{n})+\int_{t_{n}}^{t_{n+\frac{1}{2}}}P_{1}^{(i)}(t,X)dt+O(h^{3}),\text{\,\,\,for\,\,\,}i=1,2,...,9,
       \end{equation}
       which is equivalent to the following system
       \begin{equation}\label{42}
        X(t_{n+\frac{1}{2}})=X(t_{n})+\int_{t_{n}}^{t_{n+\frac{1}{2}}}Q_{1}(t,X)dt+\mathcal{O}(h^{3}),
       \end{equation}
       where $Q_{1}(t,X(t))=\left(P_{1}^{(1)}(t,X(t)),P_{1}^{(2)}(t,X(t)),...,P_{1}^{(9)}(t,X(t))\right)^{T}$ and
       $\mathcal{O}(h^{3})=(O(h^{3}),O(h^{3}),...,O(h^{3}))^{T}.$\\

       The integration of both sides of equation $(\ref{38})$ provides
       \begin{equation*}
        \int_{t_{n}}^{t_{n+\frac{1}{2}}}P_{1}^{(i)}(t,X)dt=\frac{1}{h}\left[\left[F_{i}\left(t_{n+\frac{1}{2}},X(t_{n+\frac{1}{2}})\right)-
        F_{i}(t_{n},X(t_{n}))\right](t_{n+\frac{1}{2}}^{2}-t_{n}^{2})+2\left[t_{n+\frac{1}{2}}F_{i}(t_{n},X(t_{n}))-\right.\right.
       \end{equation*}
       \begin{equation*}
        \left.\left.t_{n}F_{i}\left(t_{n+\frac{1}{2}},X(t_{n+\frac{1}{2}})\right)\right](t_{n+\frac{1}{2}}-t_{n})\right]=
        \left[F_{i}\left(t_{n+\frac{1}{2}},X(t_{n+\frac{1}{2}})\right)-F_{i}(t_{n},X(t_{n}))\right](t_{n}+\frac{h}{4})
       \end{equation*}
       \begin{equation}\label{43}
        +t_{n+\frac{1}{2}}F_{i}(t_{n},X(t_{n}))-t_{n}F_{i}\left(t_{n+\frac{1}{2}},X(t_{n+\frac{1}{2}})\right)=
        \frac{h}{4}\left[F_{i}\left(t_{n+\frac{1}{2}},X(t_{n+\frac{1}{2}})\right)+F_{i}(t_{n},X(t_{n}))\right],
       \end{equation}
       where the last two equalities come from the identities $t_{n+\frac{1}{2}}^{2}-t_{n}^{2}=(t_{n+\frac{1}{2}}-t_{n})(t_{n+\frac{1}{2}}+t_{n})$,
       $t_{n+\frac{1}{2}}-t_{n}=\frac{h}{2}$ and $t_{n+\frac{1}{2}}+t_{n}=2t_{n}+\frac{h}{2}.$ To get the desired first-level of the new algorithm,
       we should approximate the sum $F_{i}\left(t_{n+\frac{1}{2}},X(t_{n+\frac{1}{2}})\right)+F_{i}(t_{n},X(t_{n}))$ by the term
       $a_{1}F_{i}(t_{n},X(t_{n}))+a_{2}F_{i}\left(t_{n}+p_{1}h,X(t_{n})+p_{2}hF(t_{n},X(t_{n}))\right)$, in which the coefficients $a_{1}$,
       $a_{2}$, $p_{1}$ and $p_{2}$, are real numbers and are chosen so that the Taylor expansion
       \begin{equation*}
        \frac{X_{i}(t_{n+\frac{1}{2}})-X_{i}(t_{n})}{h/2}-\frac{1}{2}\left[F_{i}\left(t_{n+\frac{1}{2}},X(t_{n+\frac{1}{2}})\right)+F_{i}(t_{n},X(t_{n}))\right]
        =O(h^{2}).
       \end{equation*}
       The application of the Taylor series expansion for $X_{i}$ and $F_{i}$ about $t_{n}$ and $(t_{n},X(t_{n}))$, respectively, with step size
       $h/2$ using forward difference representations gives
       \begin{equation}\label{45}
        X_{i}(t_{n+\frac{1}{2}})=X_{i}(t_{n})+\frac{h}{2}F_{i}(t_{n},X(t_{n}))+\frac{h^{2}}{8}\left[\partial_{t}F_{i}(t_{n},X(t_{n}))
        +\underset{k=1}{\overset{9}\sum}F_{k}(t_{n},X(t_{n}))\partial_{k}F_{i}(t_{n},X(t_{n}))\right]+O(h^{3}),
       \end{equation}
       and
       \begin{equation*}
        F_{i}\left(t_{n}+p_{1}h,X(t_{n})+p_{2}hF(t_{n},X(t_{n}))\right)=F_{i}(t_{n},X(t_{n}))+hp_{1}\partial_{t}F_{i}(t_{n},X(t_{n}))+
       \end{equation*}
       \begin{equation}\label{46}
       hp_{2}\underset{k=1}{\overset{9}\sum}F_{k}(t_{n},X(t_{n}))\partial_{k}F_{i}(t_{n},X(t_{n}))+O(h^{2}),
       \end{equation}
       where $\partial_{t}F_{i}$ denotes $\frac{\partial F_{i}}{\partial t}$ and $\partial_{k}F_{i}$ represent $\frac{\partial F_{i}}{\partial X_{k}}$, for
       $k=1,2,...,9$.\\

       Combining equations $(\ref{45})$-$(\ref{46}),$ direct calculations yield
       \begin{equation*}
        \frac{X_{i}(t_{n+\frac{1}{2}})-X_{i}(t_{n})}{h/2}-\frac{1}{2}\left[F_{i}\left(t_{n+\frac{1}{2}},X(t_{n+\frac{1}{2}})\right)+F_{i}(t_{n},X(t_{n}))\right]
        =\left(1-\frac{a_{1}+a_{2}}{2}\right)F_{i}(t_{n},X(t_{n}))+
       \end{equation*}
       \begin{equation*}
       \frac{h}{2}\left[\left(\frac{1}{2}-a_{2}p_{1}\right)\partial_{t}F_{i}(t_{n},X(t_{n}))+\left(\frac{1}{2}-a_{2}p_{2}\right)\underset{k=1}{\overset{9}\sum}
        F_{k}(t_{n},X(t_{n}))\partial_{k}F_{i}(t_{n},X(t_{n}))\right]+O(h^{2}),
       \end{equation*}
       which equals $O(h^{2})$ if and only if $a_{1}+a_{2}=2,$ $a_{2}p_{1}=\frac{1}{2}$ and $a_{2}p_{2}=\frac{1}{2}$. But the last two equations require
        $a_{2}\neq0$, $p_{1}\neq0$ and $p_{2}\neq0$. Take for instance
       \begin{equation}\label{47}
        p_{1}=p_{2}=\frac{1}{2}, \text{\,\,\,so\,\,\,} a_{1}=a_{2}=1.
       \end{equation}
       Plugging equation $(\ref{46})$ and relation $(\ref{47})$, it is not hard to observe that the sum
        $F_{i}\left(t_{n+\frac{1}{2}},X(t_{n+\frac{1}{2}})\right)+F_{i}(t_{n},X(t_{n}))$ is approximate as
       \begin{equation*}
        F_{i}\left(t_{n+\frac{1}{2}},X(t_{n+\frac{1}{2}})\right)+F_{i}(t_{n},X(t_{n}))=2F_{i}(t_{n},X(t_{n}))+\frac{h}{2}
       \left[\partial_{t}F_{i}(t_{n},X(t_{n}))+\right.
       \end{equation*}
       \begin{equation}\label{48}
       \left.\underset{k=1}{\overset{9}\sum}F_{k}(t_{n},X(t_{n}))\partial_{k}F_{i}(t_{n},X(t_{n}))\right]+O(h^{2}).
       \end{equation}
       A combination of equations $(\ref{43})$ and $(\ref{48})$ results in
       \begin{equation}\label{49}
        \int_{t_{n}}^{t_{n+\frac{1}{2}}}P_{1}^{(i)}(t,X)dt=\frac{h}{2}F_{i}(t_{n},X(t_{n}))+\frac{h^{2}}{8}
        \left[\partial_{t}F_{i}(t_{n},X(t_{n}))+\underset{k=1}{\overset{9}\sum}F_{k}(t_{n},X(t_{n}))\partial_{k}F_{i}(t_{n},X(t_{n}))\right]+O(h^{3}).
       \end{equation}
       Substituting equation $(\ref{49})$ into $(\ref{41})$, this provides
       \begin{equation}\label{44}
       X_{i}(t_{n+\frac{1}{2}})=X_{i}(t_{n})+\frac{h}{2}F_{i}(t_{n},X(t_{n}))+\frac{h^{2}}{8}\left[\partial_{t}F_{i}(t_{n},X(t_{n}))
       +\underset{k=1}{\overset{9}\sum}F_{k}(t_{n},X(t_{n}))\partial_{k}F_{i}(t_{n},X(t_{n}))\right]+O(h^{3}).
       \end{equation}
       Tracking the infinitesimal term $O(h^{3}),$ equation $(\ref{44})$ can be approximate as
       \begin{equation}\label{50}
       X_{i}^{n+\frac{1}{2}}=X_{i}^{n}+\frac{h}{2}F_{i}(t_{n},X^{n})+\frac{h^{2}}{8}\left[\partial_{t}F_{i}(t_{n},X^{n})+\underset{k=1}{\overset{9}\sum}
        F_{k}(t_{n},X^{n})\partial_{k}F_{i}(t_{n},X^{n})\right],\text{\,\,\,for\,\,\,}i=1,2,...,9.
       \end{equation}
       The difference equations provided by relation $(\ref{50})$ represent the first-level of the new approach.\\

       To develop the second-level of the desired algorithm, we should integrate both sides of system $(\ref{s1})$ at the node points $t_{n+\frac{1}{2}}$
       and $t_{n+1}$. Thus
       \begin{equation*}
      X(t_{n+1})-X(t_{n+\frac{1}{2}})=\int_{t_{n+\frac{1}{2}}}^{t_{n+1}}F(t,X)dt,
      \end{equation*}
       which can be rewritten as
       \begin{equation*}
      X(t_{n+1})=X(t_{n+\frac{1}{2}})+\int_{t_{n+\frac{1}{2}}}^{t_{n+1}}F(t,X)dt.
      \end{equation*}
       This implies
       \begin{equation}\label{51}
        X_{i}(t_{n+1})=X_{i}(t_{n+\frac{1}{2}})+\int_{t_{n+\frac{1}{2}}}^{t_{n+1}}F_{i}(t,X)dt,\text{\,\,\,for\,\,\,}i=1,2,...,9.
        \end{equation}
       Replacing $F_{i}(t,X)$ by the linear interpolating polynomial $P_{1}^{(i)}(t,X)$ at the node points $(t_{n},F_{i}(t_{n},X^{n}))$ and
       $(t_{n+\frac{1}{2}},F_{i}(t_{n+\frac{1}{2}},X^{n})),$ and using $(\ref{39})$, equation $(\ref{51})$ is approximate as
       \begin{equation}\label{52}
        X_{i}(t_{n+1})=X_{i}(t_{n+\frac{1}{2}})+\int_{t_{n+\frac{1}{2}}}^{t_{n+1}}P_{1}^{(i)}(t,X)dt+O(h^{3}).
        \end{equation}
       Plugging $(\ref{52})$ and $(\ref{38})$, simple computations give
       \begin{equation*}
        X_{i}(t_{n+1})=X_{i}(t_{n+\frac{1}{2}})+\frac{2}{h}\left[\frac{1}{2}\left[F_{i}\left(t_{n+\frac{1}{2}},X(t_{n+\frac{1}{2}})\right)-
        F_{i}(t_{n},X(t_{n}))\right](t_{n+1}^{2}-t_{n+\frac{1}{2}}^{2})+\left[t_{n+1}F_{i}(t_{n},X(t_{n}))-\right.\right.
       \end{equation*}
       \begin{equation*}
        \left.\left.t_{n+\frac{1}{2}}F_{i}\left(t_{n+\frac{1}{2}},X(t_{n+\frac{1}{2}})\right)\right](t_{n+1}-t_{n+\frac{1}{2}})\right]+O(h^{3})=
        X_{i}(t_{n+\frac{1}{2}})+
       \end{equation*}
       \begin{equation}\label{52a}
        \frac{h}{2}\left[3F_{i}\left(t_{n+\frac{1}{2}},X(t_{n+\frac{1}{2}})\right)-F_{i}(t_{n},X(t_{n}))\right]+O(h^{3}).
       \end{equation}
       Omitting the error term $O(h^{3}),$ we obtain the second-level of the new method which is defined as
       \begin{equation}\label{53}
        X_{i}^{n+1}=X_{i}^{n+\frac{1}{2}}+\frac{h}{2}\left[3F_{i}\left(t_{n+\frac{1}{2}},X^{n+\frac{1}{2}}\right)
        -F_{i}(t_{n},X^{n})\right],\text{\,\,\,for\,\,\,}i=1,2,...,9.
       \end{equation}
       Under the assumption given by equations $(\ref{19a})$, $(\ref{28})$,$(\ref{29a})$,$(\ref{30})$ and $(\ref{33})$, the functions $F_{i}(t,X)$,
       ($i=1,2,...,9,$) defined by relations $(\ref{10})$-$(\ref{14})$ becomes
       \begin{equation}\label{530}
       F_{1}(t,X(t))=-\frac{m(t)X_{1}}{N}\left[\beta_{X_{3}}(t)X_{2}+\beta_{X_{3}}(t)X_{3}+\beta_{X_{4}}(\theta)X_{2}+
       C_{X_{5}}(t)\beta_{X_{3}}(t)X_{5}+C_{X_{5}}(t)\beta_{X_{3}}(t)X_{6}\right],
      \end{equation}
       \begin{equation}\label{531}
       F_{2}(t,X(t))=\frac{m(t)X_{1}}{N}\left[\beta_{X_{3}}(t)X_{2}+\beta_{X_{3}}(t)X_{3}+\beta_{X_{4}}(\theta)X_{2}+
       C_{X_{5}}(t)\beta_{X_{3}}(t)X_{5}+C_{X_{5}}(t)\beta_{X_{3}}(t)X_{6}\right]-\alpha X_{2},
      \end{equation}
       \begin{equation}\label{532}
       F_{3}(t,X(t))=\alpha X_{2}-\gamma(t)X_{3},\text{\,\,\,}F_{4}(t,X(t))=(1-\theta(t))\gamma(t)X_{3}-\rho(t)X_{4},
      \end{equation}
      \begin{equation}\label{533}
       F_{5}(t,X(t))=(\theta(t)-w(t))\gamma(t)(t)X_{3}-\rho(t)X_{5},\text{\,\,\,}F_{6}(t,X(t))=w(t)\gamma(t)X_{3}-\psi(t)X_{6},
      \end{equation}
      \begin{equation}\label{534}
       F_{7}(t,X(t))=\rho(t)X_{5},\text{\,\,\,}F_{8}(t,X(t))=\rho(t)X_{4}\text{\,\,\,and\,\,\,}F_{9}(t,X(t))=\psi(t)X_{6},
      \end{equation}
      By straightforward computations, it is not hard to observe that
      \begin{equation*}
      \partial_{t}F_{1}=\frac{-1}{N}\left\{C_{X_{2}}\beta_{X_{2}}\left[(m(t)F_{1}+\dot{m}(t)X_{1})X_{2}+m(t)X_{1}F_{2}\right]+
      \beta\left[(m(t)F_{1}+\dot{m}(t)X_{1})X_{3}+m(t)X_{1}F_{3}\right]\right.
      \end{equation*}
      \begin{equation*}
      +\left[m(t)\beta_{X_{4}}(\theta)F_{1}+(\dot{m}(t)\beta_{X_{4}}(\theta)+m(t)\dot{\beta}_{X_{4}}(\theta))X_{1}\right]X_{4}+m(t)
      \beta_{X_{4}}(\theta)X_{1}F_{4}+\beta\left[\left(m(t)c_{X_{10}}(t)F_{1}+\right.\right.
      \end{equation*}
      \begin{equation}\label{54}
      \left.\left.\left.\left(\dot{m}(t)c_{X_{10}}(t)+m(t)\dot{c}_{X_{10}}(t)\right)X_{1}\right)(X_{5}+X_{6})+m(t)c_{X_{10}}(t)
      X_{1}(F_{5}+F_{6})\right]\right\},
      \end{equation}
      \begin{equation}\label{55}
      \partial_{t}F_{2}=-\partial_{t}F_{1}-\dot{\gamma}(t)X_{_{2}}-\gamma(t)F_{_{2}},\text{\,\,\,\,}\partial_{t}F_{3}=\alpha F_{2}
      -\dot{\gamma}(t)X_{_{3}}-\gamma(t)F_{_{3}},
      \end{equation}
       \begin{equation}\label{56}
      \partial_{t}F_{4}=\left[(1-\theta(t))\dot{\gamma}(t)-\dot{\theta}(t)\gamma(t)\right]X_{3}+(1-\theta(t))\gamma(t)F_{3}-\dot{\rho}(t)X_{_{4}}
      -\rho(t)F_{_{4}},
      \end{equation}
      \begin{equation}\label{57}
      \partial_{t}F_{5}=\left[(\dot{\theta}(t)-\dot{w}(t))\gamma(t)+(\theta(t)-w(t))\dot{\gamma}(t)\right]X_{3}+(\theta(t)-w(t))\gamma(t)F_{3}
       -\dot{\rho}(t)X_{_{5}}-\rho(t)F_{_{5}},
      \end{equation}
       \begin{equation}\label{58}
      \partial_{t}F_{6}=\left(\dot{w}(t)\gamma(t)+w(t)\dot{\gamma}(t)\right)X_{3}+w(t)\gamma(t)F_{3}-\dot{\psi}(t)X_{_{6}}-\psi(t)F_{_{6}},
      \end{equation}
       \begin{equation}\label{59}
      \partial_{t}F_{7}=\dot{\rho}(t)X_{_{5}}+\rho(t)F_{_{5}},\text{\,\,\,\,}\partial_{t}F_{8}=\dot{\rho}(t)X_{_{4}}+\rho(t)F_{_{4}},\text{\,\,\,\,}
      \partial_{t}F_{9}=\dot{\psi}(t)X_{_{6}}+\psi(t)F_{_{6}},
      \end{equation}
      \begin{equation}\label{59}
      \partial_{1}F_{1}=\frac{-m(t)}{N}\left\{C_{X_{2}}\beta X_{2}+\beta X_{3}+\beta_{X_{4}}(\theta)X_{4}+C_{X_{10}}(t)\beta(X_{5}+X_{6})\right\},
      \end{equation}
      \begin{equation}\label{60}
      \partial_{2}F_{1}=\frac{-m(t)}{N}C_{X_{2}}\beta X_{1},\text{\,\,}\partial_{3}F_{1}=\frac{-m(t)}{N}\beta X_{1},\text{\,\,}
      \partial_{4}F_{1}=\frac{-m(t)}{N}\beta_{X_{4}}(\theta)X_{1},
      \end{equation}
      \begin{equation}\label{61}
      \partial_{5}F_{1}=\partial_{6}F_{1}=\frac{-m(t)}{N}C_{X_{10}}\beta X_{1},\text{\,\,}\partial_{k}F_{1}=0,\text{\,\,\,for\,\,\,}k=7,8,9,
      \end{equation}
      \begin{equation}\label{62}
      \partial_{1}F_{2}=-\partial_{1}F_{1},\text{\,\,\,}\partial_{2}F_{2}=-\partial_{2}F_{1}-\gamma(t),\text{\,\,\,}
      \partial_{k}F_{2}=0,\text{\,\,\,for\,\,\,}k=3,4,...,9,
      \end{equation}
      \begin{equation}\label{63}
      \partial_{1}F_{3}=0,\text{\,\,\,}\partial_{2}F_{3}=\alpha,\text{\,\,\,}\partial_{3}F_{3}=-\gamma(t),\text{\,\,\,}
      \partial_{k}F_{3}=0,\text{\,\,\,for\,\,\,}k=4,5,...,9,
      \end{equation}
      \begin{equation}\label{64}
      \partial_{3}F_{3}=(1-\theta(t))\gamma(t),\text{\,\,\,}\partial_{4}F_{4}=-\rho(t),\text{\,\,\,}\partial_{k}F_{4}=0,\text{\,\,\,for\,\,\,}
      k\neq 3,4,
      \end{equation}
      \begin{equation}\label{65}
      \partial_{3}F_{5}=(\theta(t)-w(t))\gamma(t),\text{\,\,\,}\partial_{5}F_{5}=-\rho(t),\text{\,\,\,}\partial_{k}F_{5}=0,\text{\,\,\,for\,\,\,}
      k\neq 3,5,
      \end{equation}
      \begin{equation}\label{66}
      \partial_{3}F_{6}=w(t)\gamma(t),\text{\,\,\,}\partial_{6}F_{6}=-\psi(t),\text{\,\,\,}\partial_{k}F_{6}=0,\text{\,\,\,for\,\,\,}
      k\neq 3,6,\text{\,\,}\partial_{5}F_{7}=\rho(t),\text{\,\,}\partial_{k}F_{7}=0,\text{\,\,for\,\,}k\neq5,
      \end{equation}
      \begin{equation}\label{67}
      \partial_{4}F_{8}=\rho(t),\text{\,\,}\partial_{k}F_{8}=0,\text{\,\,\,for\,\,\,}k\neq4,\text{\,\,}
      \partial_{6}F_{9}=\psi(t),\text{\,\,}\partial_{k}F_{9}=0,\text{\,\,for\,\,}k\neq6.
      \end{equation}
      To provide a full description of the two-level explicit formulation for solving the mathematical problem $(\ref{1})$-$(\ref{5})$ to predicting the
      spread of Covid-19 model, we should put together relations $(\ref{50})$, $(\ref{53})$ and the initial condition given by equation $(\ref{6})$.
       Specifically, for $n=1,2,...,M-1,$
      \begin{equation}\label{e1}
       X_{i}^{n+\frac{1}{2}}=X_{i}^{n}+\frac{h}{2}F_{i}(t_{n},X^{n})+\frac{h^{2}}{8}\left[\partial_{t}F_{i}(t_{n},X^{n})+\underset{k=1}{\overset{9}\sum}
        F_{k}(t_{n},X^{n})\partial_{k}F_{i}(t_{n},X^{n})\right],\text{\,\,\,for\,\,\,}i=1,2,...,9;
       \end{equation}
       \begin{equation}\label{e2}
        X_{i}^{n+1}=X_{i}^{n+\frac{1}{2}}+\frac{h}{2}\left[3F_{i}\left(t_{n+\frac{1}{2}},X^{n+\frac{1}{2}}\right)
        -F_{i}(t_{n},X^{n})\right],\text{\,\,\,for\,\,\,}i=1,2,...,9;
       \end{equation}
        subject to the initial condition
        \begin{equation}\label{e3}
      X_{i}(t_{0})=X_{i}^{0},\text{\,\,\,for\,\,\,}i=1,2,...,9;
      \end{equation}
      where the functions $F_{i}$ ($i=1,2,...,9$) and its partial derivatives are given by equations $(\ref{530})$-$(\ref{67}).$\\

      Using the tools provided in section $\ref{sec3}$ we are ready to analyze the stability and the convergence rate of the two-level explicit procedure
      $(\ref{e1})$-$(\ref{e3})$ for predicting the spread of SARS-Cov-2 epidemic modeled by equations $(\ref{1})$-$(\ref{6})$.

      \section{Stability analysis and convergence rate of the new algorithm}\label{sec4}
      In this section we wish to examine the stability and convergence rate of the new technique $(\ref{e1})$-$(\ref{e3})$ applied to the initial-value
        problem $(\ref{1})$-$(\ref{6}).$\\

      Firstly, we define the strip $\mathcal{S}=\{(t,X),\text{\,\,}t_{0}\leq t\leq T_{max},\text{\,\,}X\in\mathbb{R}^{9}\}$ in which both exact and
      computed solutions of problem $(\ref{1})$-$(\ref{6})$ should lie. It comes from equations $(\ref{530})$-$(\ref{67})$ that the functions $F_{i}$
      and their partial
        derivatives are continuous on the strip $\mathcal{S}$, but the partial derivatives are unbounded on this set. Thus it follows from the Henrici
       result \cite{he1962} that the system of equations $(\ref{1})$-$(\ref{6})$ admits a unique solution $X(t)$ defined in a certain neighborhood
       $U(t_{0})\subset[t_{0},T_{max}]$ of the initial point $t_{0}.$ Without loss of this constraint, we assume in the following that
       $U(t_{0})=[t_{0},T_{max}]$ (indeed, we are dealing with a real world problem which, in reality should have a unique solution defined over the
        interval $[t_{0},T_{max}]$). This shows the existence and uniqueness of the solution of the initial-value problem $(\ref{1})$-$(\ref{6})$.\\

      Let now introduce the functions $\Delta_{i}(t_{l},X_{i}(t_{l}))$ and $\delta_{i}(t_{l},X_{i}(t_{l}))$ (for $l=n,n+\frac{1}{2}$) be the difference
      quotient of the exact solution $X_{i}(t_{l})$ of equations $(\ref{1})$-$(\ref{6})$ at time $t_{l}$ and the difference quotient of the approximate
       solution $X_{i}^{l}$ of problem $(\ref{e1})$-$(\ref{e3})$ obtained at time $t_{l}$, respectively. Moreover, $\Delta_{i}(t_{l},X_{i}(t_{l}))$ and
       $\delta_{i}(t_{l},X_{i}(t_{l}))$ are given by
      \begin{equation}\label{81}
        \Delta_{i}(t_{l},X_{i}(t_{l}))=\left\{
                                         \begin{array}{ll}
                                           \frac{X_{i}(t_{l}+\frac{1}{2})-X_{i}(t_{l})}{h/2}, & \hbox{if $h\neq0,$} \\
                                           \text{\,}\\
                                           F_{i}(t_{l},X_{i}(t_{l})), & \hbox{for $h=0$,}
                                         \end{array}
                                       \right.
      \end{equation}
       and
      \begin{equation}\label{82}
        \delta_{i}(t_{l},X_{i}(t_{l}))=\left\{
                                         \begin{array}{ll}
                                           \frac{X_{i}^{l+\frac{1}{2}}-X_{i}^{l}}{h/2}, & \hbox{if $h\neq0,$} \\
                                            \text{\,}\\
                                           F_{i}(t_{l},X_{i}^{l}), & \hbox{for $h=0$.}
                                         \end{array}
                                       \right.
      \end{equation}
       It is worth mentioning that $\delta_{i}(t_{l},X_{i}(t_{l}))=\delta_{i}(t_{l},X_{i}^{l}).$ The local discretization error at the point
       $(t_{l},X_{i}(t_{l}))$ of the considered scheme is defined as
       \begin{equation}\label{83}
        \sigma_{i}(t_{l},X_{i}(t_{l}))=\Delta_{i}(t_{l},X_{i}(t_{l}))-\delta_{i}(t_{l},X_{i}(t_{l})),\text{\,\,\,}l=n\text{\,\,or\,\,}n+\frac{1}{2},
        \text{\,\,\,for\,\,\,}i=1,2,...,9,
       \end{equation}
        indicates how well the exact solution of the differential equations $(\ref{1})$-$(\ref{6})$ obeys the difference equations
        $(\ref{e1})$-$(\ref{e3})$ provided by the two-level explicit formulation.\\

        The following result (Theorem $\ref{t1}$) analyzes the stability and gives the convergence rate of the proposed approach $(\ref{e1})$-$(\ref{e3}).$

      \begin{theorem}\label{t1}(Stability analysis and convergence rate)\\
      Let $e^{n}=X(t_{n})-X^{n}$ be the global discretization error provided by algorithm $(\ref{e1})$-$(\ref{e3})$, where $X(t_{n})$ is the solution of
      system $(\ref{1})$-$(\ref{6})$ obtained at time $t_{n}$ and $X^{n}$ is the one provided by $(\ref{e1})$-$(\ref{e3})$ at time $t_{n}$. Thus, it holds
       \begin{equation*}
      \|X^{n}\|_{\infty}=\underset{1\leq 1\leq 9}{\max}|X_{i}^{n}|\leq C_{1},\text{\,\,which\,\,implies\,\,}\||X|\|_{L^{2}(I)}=
      \left(h\underset{n=0}{\overset{M}\sum}\|X^{n}\|_{\infty}^{2}\right)^{\frac{1}{2}}\leq C_{1}
      \end{equation*}
      where $C_{1}$ is a positive constant independent of the step size $h$ and $X$ denotes the approximate solution. Furthermore
       \begin{equation*}
      \|e^{n}\|_{\infty}=\underset{1\leq 1\leq 9}{\max}|e_{i}^{n}|\leq C_{2}h^{2},\text{\,\,which\,\,implies\,\,}
      \||e|\|_{L^{2}(I)}=\left(h\underset{n=0}{\overset{M}\sum}\|e^{n}\|_{\infty}^{2}\right)^{\frac{1}{2}}\leq C_{2}h^{2}
      \end{equation*}
      where $C_{2}$ is a positive constant that does not depend on the step size $h$. In the following we represent the analytical solution by
      $X(\cdot).$\\
      \end{theorem}

      The proof of Theorem $\ref{t1}$ requires the following intermediate result (namely Lemmas $\ref{l1}$).

       \begin{lemma}\label{l1}
       Suppose the vectors $q_{j}\in\mathbb{C}^{l}$ satisfy the estimates of the form
       \begin{equation}\label{88}
        \|q_{j+1}\|_{L^{\infty}(\mathbb{C}^{l})}\leq(1+\epsilon)\|q_{j}\|_{L^{\infty}(\mathbb{C}^{l})}+\xi,\text{\,\,\,}\epsilon>0\text{\,\,and\,\,}\xi>0,
       \end{equation}
       for $j=0,1,...,m.$ Then
        \begin{equation}\label{85}
        \|q_{m}\|_{L^{\infty}(\mathbb{C}^{l})}\leq\exp(m\epsilon)\|q_{0}\|_{L^{\infty}(\mathbb{C}^{l})}+\frac{\exp(m\epsilon)-1}{\epsilon}\xi.
       \end{equation}
       \end{lemma}

       \begin{proof}
       We should prove inequality $(\ref{85})$ by mathematical induction. Since $\exp(\epsilon)\geq 1+\epsilon$ and
       $1\leq\frac{\exp(\epsilon)-1}{\epsilon},$ for any $\epsilon\geq0,$ using the assumption of Lemma $\ref{l1}$, it is not difficult to see that
         \begin{equation*}
        \|q_{1}\|_{L^{\infty}(\mathbb{C}^{l})}\leq(1+\epsilon)\|q_{0}\|_{L^{\infty}(\mathbb{C}^{l})}+\xi\leq
        \exp(\epsilon)\|q_{0}\|_{L^{\infty}(\mathbb{C}^{l})}+\frac{\exp(\epsilon)-1}{\epsilon}\xi.
         \end{equation*}
      Now, let assume that
       \begin{equation}\label{87}
        \|q_{m-1}\|_{L^{\infty}(\mathbb{C}^{l})}\leq\exp((m-1)\epsilon)\|q_{0}\|_{L^{\infty}(\mathbb{C}^{l})}+\frac{\exp((m-1)\epsilon)-1}{\epsilon}\xi.
       \end{equation}
       Combining estimates: $1+\epsilon\leq\exp(\epsilon)$ and $1\leq\frac{\exp(\epsilon)-1}{\epsilon}$ together with inequality $(\ref{88})$ provided
       by the assumption of Lemma $\ref{l1}$ and $(\ref{87})$, direct calculations give
        \begin{equation*}
        \|q_{m}\|_{L^{\infty}(\mathbb{C}^{l})}\leq(1+\epsilon)\|q_{m-1}\|_{L^{\infty}(\mathbb{C}^{l})}+\xi\leq(1+\epsilon)
       \left[\exp((m-1)\epsilon)\|q_{0}\|_{L^{\infty}(\mathbb{C}^{l})}+\frac{\exp((m-1)\epsilon)-1}{\epsilon}\xi\right]+\xi\leq
        \end{equation*}
         \begin{equation*}
        \exp(\epsilon)\left[\exp((m-1)\epsilon)\|q_{0}\|_{L^{\infty}(\mathbb{C}^{l})}+\frac{\exp((m-1)\epsilon)-1}{\epsilon}\xi\right]+\xi
        \leq\exp(m\epsilon)\|q_{0}\|_{L^{\infty}(\mathbb{C}^{l})}+\frac{\exp(m\epsilon)-1}{\epsilon}\xi.
        \end{equation*}
         The last inequality comes the estimate $\epsilon-\exp(\epsilon)\leq-1.$ This completes the proof of Lemma $\ref{l1}$.
       \end{proof}
        Using Lemma $\ref{l1}$, we are ready to prove the main result of this paper (namely Theorem $\ref{t1}$).

         \begin{proof} (of Theorem $\ref{t1}$).\\
       First of all, introduce the domains
       \begin{equation*}
       D_{i}=\{(t,x):\text{\,\,}t\in[t_{0},T_{max}],\text{\,\,}x\in\mathbb{R},\text{\,\,}|X_{i}(t)-x|\leq \upsilon\},\text{\,\,\,}
       \mathcal{S}_{0}=\{(t,x):\text{\,\,}t_{0}\leq t\leq T_{max},\text{\,\,}x\in\mathbb{R}\},
       \end{equation*}
        and
       \begin{equation}\label{D1}
       D=\{(t,X):\text{\,\,}t\in[t_{0},T_{max}],\text{\,\,}X\in\mathbb{R}^{9},\text{\,\,}\|X(t)-X\|_{\infty}\leq \upsilon\},
       \end{equation}
       where $\upsilon$ is a positive constant independent of the step size $h$ and $X_{i}(t)$ denotes the ith component of the exact solution $X(t)$
       of the initial-value problem $(\ref{1})$-$(\ref{6})$.\\
       Plugging approximations $(\ref{e1})$-$(\ref{e2})$ and equation $(\ref{82})$, it is not hard to observe that
         \begin{equation}\label{89}
        \delta_{i}(t_{n},X_{i}(t_{n}))=F_{i}(t_{n},X^{n})+\frac{h}{4}\left[\partial_{t}F_{i}(t_{n},X^{n})+\underset{k=1}{\overset{9}\sum}
        F_{k}(t_{n},X^{n})\partial_{k}F_{i}(t_{n},X^{n})\right],
        \end{equation}
         and
        \begin{equation}\label{90}
        \delta_{i}(t_{n+\frac{1}{2}},X_{i}(t_{n+\frac{1}{2}}))=3F_{i}(t_{n+\frac{1}{2}},X^{n+\frac{1}{2}})-F_{i}(t_{n},X^{n}).
        \end{equation}
        Hence, the functions $\delta_{i}$ given by equation $(\ref{82})$ and their partial derivatives are continuous on the domain $D_{i}.$ Since $D_{i}$
        is a compact subset of the strip $\mathcal{S}_{0}$, their partial derivatives are bounded on $D_{i}$. Applying the Mean-Value Theorem, there exists
        a positive constant $L_{i}$ which does not depend on the step size $h$ so that
          \begin{equation}\label{83b}
           |\delta_{i}(t,X_{i})-\delta_{i}(t,Y_{i})|\leq L_{i}|X_{i}-Y_{i}|,
          \end{equation}
         for every $(t,X_{i})$ and $(t,Y_{i})$ in $D_{i}$. Furthermore, a simple manipulation of $(\ref{44})$ results in
          \begin{equation*}
        \frac{X_{i}(t_{n+\frac{1}{2}})-X_{i}(t_{n})}{h/2}=F_{i}(t_{n},X(t_{n}))+\frac{h}{4}\left[\partial_{t}F_{i}(t_{n},X(t_{n}))
       +\underset{k=1}{\overset{9}\sum}F_{k}(t_{n},X(t_{n}))\partial_{k}F_{i}(t_{n},X(t_{n}))\right]+O(h^{2}),
          \end{equation*}
       which is equivalent to
       \begin{equation*}
        \frac{X_{i}(t_{n+\frac{1}{2}})-X_{i}(t_{n})}{h/2}-\left[F_{i}(t_{n},X(t_{n}))+\frac{h}{4}\left(\partial_{t}F_{i}(t_{n},X(t_{n}))
       +\underset{k=1}{\overset{9}\sum}F_{k}(t_{n},X(t_{n}))\partial_{k}F_{i}(t_{n},X(t_{n}))\right)\right]=O(h^{2}).
          \end{equation*}
        Utilizing equations $(\ref{81})$, $(\ref{83})$ and $(\ref{89})$, this becomes
        \begin{equation*}
         \sigma_{i}(t_{n},X_{i}(t_{n}))=\Delta_{i}(t_{n},X_{i}(t_{n}))-\delta_{i}(t_{n},X_{i}(t_{n}))=O(h^{2}),
        \end{equation*}
          which implies
         \begin{equation}\label{91}
         |\sigma_{i}(t_{n},X_{i}(t_{n}))|\leq C_{i4}h^{2},\text{\,\,\,for\,\,\,}i=1,2,...,9,
          \end{equation}
        where $C_{i4}$ are positive constants that do not depend on the step size $h.$ Setting $C_{4}=\underset{1\leq i\leq9}{\max}C_{i4}$;
         $\sigma(t_{n},X(t_{n}))=(\sigma_{1}(t_{n},X_{1}(t_{n})),...,\sigma_{9}(t_{n},X_{9}(t_{n})))^{T}$ and taking the maximum of both
         sides of estimate $(\ref{91})$ over $i$, to get
           \begin{equation}\label{92}
         \|\sigma(t_{n},X_{i}(t_{n}))\|_{L^{\infty}}\leq C_{4}h^{2}.
          \end{equation}
         In a similar manner, combining approximation $(\ref{53})$, equations $(\ref{81})$, $(\ref{82})$ and $(\ref{83})$, straightforward calculations
         provide
       \begin{equation*}
        \frac{X_{i}(t_{n+1})-X_{i}(t_{n+\frac{1}{2}})}{h/2}-\left[3F_{i}(t_{n+\frac{1}{2}},X(t_{n+\frac{1}{2}}))-F_{i}(t_{n},X(t_{n}))\right]=O(h^{2}),
        \end{equation*}
        which can be rewritten as
        \begin{equation*}
         \Delta_{i}(t_{n+\frac{1}{2}},X_{i}(t_{n+\frac{1}{2}}))-\delta_{i}(t_{n+\frac{1}{2}},X_{i}(t_{n+\frac{1}{2}}))=O(h^{2}).
        \end{equation*}
         Utilizing the definition of $\sigma_{i}(t_{n+\frac{1}{2}},X_{i}(t_{n+\frac{1}{2}})),$ this implies
          \begin{equation*}
         |\sigma_{i}(t_{n+\frac{1}{2}},X_{i}(t_{n+\frac{1}{2}}))|\leq C_{i5}h^{2},
        \end{equation*}
          where $C_{i5}$ ($1\leq i\leq9$) are positive constants independent of $h.$ Taking the maximum over $i,$ this yields
         \begin{equation}\label{93}
         \|\sigma(t_{n+\frac{1}{2}},X(t_{n+\frac{1}{2}}))\|_{L^{\infty}}\leq C_{5}h^{2},
          \end{equation}
        where $C_{5}=\underset{1\leq i\leq9}{\max}C_{i5}$; $\sigma(t_{n+\frac{1}{2}},X(t_{n+\frac{1}{2}}))=
         (\sigma_{1}(t_{n+\frac{1}{2}},X_{1}(t_{n+\frac{1}{2}})),...,\sigma_{9}(t_{n+\frac{1}{2}},X_{9}(t_{n+\frac{1}{2}})))^{T}$.\\

         We consider the functions $\widehat{\delta}_{i}(t,x)$ defined in the strip $\mathcal{S}_{0}=\{(t,x):\text{\,\,}t_{0}\leq t\leq
         T_{max},\text{\,\,}x\in\mathbb{R}\},$ as
           \begin{equation}\label{94}
          \widehat{\delta}_{i}(t,x)=\left\{
                                      \begin{array}{ll}
                                        \delta_{i}(t,x), & \hbox{if $(t,x)\in D_{i}$,} \\
                                        \text{\,}\\
                                        \delta_{i}(t,X_{i}(t)+\upsilon), & \hbox{for $t\in[t_{0},T_{\max}]$ and $x>X_{i}(t)+\upsilon,$} \\
                                         \text{\,}\\
                                        \delta_{i}(t,X_{i}(t)-\upsilon), & \hbox{if $t\in[t_{0},T_{\max}]$ and $x<X_{i}(t)-\upsilon.$}
                                      \end{array}
                                    \right.
           \end{equation}
            We remind that $X_{i}(t)$ is the ith component of the exact solution $X(t)$ of the initial-value problem $(\ref{1})$-$(\ref{6})$. Now, using
          relation $(\ref{94})$, it is not hard to observe that each function $\widehat{\delta}_{i}$ satisfies the "Lipschitz requirement" in the strip
          $\mathcal{S}_{0}$, that is ,
           \begin{equation}\label{94a}
           |\widehat{\delta}_{i}(t,x)-\widehat{\delta}_{i}(t,y)|\leq L_{i}|x-y|,
          \end{equation}
           for all $(t,x)$ and $(t,x)$ in $\mathcal{S}_{0}$, where the positive constant is given by relation $(\ref{83b})$. Thus, each
           $\widehat{\delta}_{i}$ is continuous on $\mathcal{S}_{0}$. Indeed. Let $(t,x)$ and $(t,y)$ be two elements of $\mathcal{S}_{0}$. If $(t,x)$
          and $(t,y)$ lie in $D_{i},$ then estimate  $(\ref{94a})$ holds thanks to inequality $(\ref{83b})$, so the function $\widehat{\delta}_{i}$ is
          continuous and satisfies the "Lipschitz requirement" on $D_{i}.$ Otherwise, either $(t,x)$ or $(t,y)$ do not lie in $D_{i}$. This corresponds
          to three
          cases: (a) $(t,x)$ lies in $D_{i}$ and $(t,y)$ does not lie in $D_{i}$, (b) ($(t,x)$ does not lie in $D_{i}$ and $(t,y)$ does lies in $D_{i}$)
          and (c) $(t,x)$ and $(t,y)$ do not lie in $D_{i}$. Here we should prove only one case, for instance $(t,x)$ does not lie in $D_{i}$ and $(t,y)$
           does lie in $D_{i}$, the proof of the two other cases are similar.\\
           \begin{equation*}
          (t,x)\notin D_{i}\Leftrightarrow|x-X_{i}(t)|>\upsilon \Leftrightarrow x-X_{i}(t)>\upsilon\text{\,\,or\,\,}-(x-X_{i}(t))>\upsilon
          \Leftrightarrow x>X_{i}(t)+\upsilon\text{\,\,or\,\,}x<X_{i}(t))-\upsilon.
           \end{equation*}
         So,
         \begin{equation*}
          |\widehat{\delta}_{i}(t,x)-\widehat{\delta}_{i}(t,y)|=\left\{
                                      \begin{array}{ll}
                                        |\delta_{i}(t,X_{i}(t)+\upsilon)-\delta_{i}(t,y)|, & \hbox{for $t_{0}\leq t\leq T_{\max}$ and
                                      $x>X_{i}(t)+\upsilon,$}\\
                                         \text{\,}\\
                                       |\delta_{i}(t,X_{i}(t)-\upsilon)-\delta_{i}(t,y)|, & \hbox{if $t\in[t_{0},T_{\max}]$ and $x<X_{i}(t)-\upsilon,$}
                                      \end{array}
                                    \right.
           \end{equation*}
           \begin{equation}\label{95}
          \leq                   \left\{
                                      \begin{array}{ll}
                                        L_{i}|X_{i}(t)+\upsilon-y|, & \hbox{if $t\in[t_{0},T_{\max}]$ and $x>X_{i}(t)+\upsilon,$}\\
                                         \text{\,}\\
                                       L_{i}|X_{i}(t)-\upsilon-y|, & \hbox{for $t_{0}\leq t\leq T_{\max}$ and $x<X_{i}(t)-\upsilon.$}
                                      \end{array}
                                    \right.
           \end{equation}
           But
            \begin{equation*}
        (t,y)\in D_{i}\Leftrightarrow|X_{i}(t)-y|\leq\upsilon\Leftrightarrow X_{i}(t)-y\leq\upsilon\text{\,\,\,\,and\,\,\,\,}-X_{i}(t))+y\leq\upsilon.
           \end{equation*}
            Using this, it is easy to see that
             \begin{equation*}
          x>X_{i}(t)+\upsilon\Leftrightarrow x-y> X_{i}(t)+\upsilon-y\geq0\Leftrightarrow |X_{i}(t)-y|> |X_{i}(t)+\upsilon-y|,
           \end{equation*}
          and
         \begin{equation*}
          x<X_{i}(t)-\upsilon\Leftrightarrow x-y<X_{i}(t)-\upsilon-y\leq0 \Leftrightarrow |X_{i}(t)-y|> |X_{i}(t)-\upsilon-y|.
           \end{equation*}
          This fact together with estimate $(\ref{95})$ results in
           \begin{equation*}
           |\widehat{\delta}_{i}(t,x)-\widehat{\delta}_{i}(t,y)|\leq L_{i}|x-y|.
           \end{equation*}
           This ends the proof of the first case (a). Thus, $\widehat{\delta}_{i}$ satisfies the "Lipschitz condition" on the strip
        $\mathcal{S}_{0}=\{(t,x):\text{\,\,}t_{0}\leq t\leq T_{max},\text{\,\,}x\in\mathbb{R}\}.$ In addition, $\widehat{\delta}_{i}$ is also
         continuous on $\mathcal{S}_{0}.$\\

         Since, $(t,X_{i}(t))\in D_{i}$, so $\widehat{\delta}_{i}(t,X_{i}(t))=\delta_{i}(t,X_{i}(t))$. A combination of $(\ref{92})$-$(\ref{93})$
          provides
         \begin{equation}\label{96a}
           |\Delta_{i}(t_{n},X_{i}(t_{n}))-\widehat{\delta}_{i}(t_{n},X_{i}(t_{n}))|=|\Delta_{i}(t_{n},X_{i}(t_{n}))-\delta_{i}(t_{n},X_{i}(t_{n}))|
            \leq C_{4}h^{2},
          \end{equation}
          and
         \begin{equation}\label{96b}
           |\Delta_{i}(t_{n+\frac{1}{2}},X_{i}(t_{n+\frac{1}{2}}))-\widehat{\delta}_{i}(t_{n+\frac{1}{2}},X_{i}(t_{n+\frac{1}{2}}))|=
          |\Delta_{i}(t_{n+\frac{1}{2}},X_{i}(t_{n+\frac{1}{2}}))-\delta_{i}(t_{n+\frac{1}{2}},X_{i}(t_{n+\frac{1}{2}}))|\leq C_{5}h^{2}.
          \end{equation}
           We recall that the two-level explicit method $(\ref{e1})$-$(\ref{e3})$ is generated by the function $\delta_{i}.$ In fact, plugging equations
           $(\ref{e1})$, $(\ref{e2})$ and $(\ref{82})$, $\delta_{i}$ is explicitly defined as
            \begin{equation*}
            \delta_{i}(t_{n},X_{i}(t_{n}))=\left\{
                                             \begin{array}{ll}
                                               F_{i}(t_{n},X^{n})+\frac{h}{4}\left(\partial_{t}F_{i}(t_{n},X^{n})+\underset{k=1}
                                     {\overset{9}\sum}F_{k}(t_{n},X^{n})\partial_{k}F_{i}(t_{n},X^{n})\right), & \hbox{if $h\neq0$} \\
                                               F_{i}(t_{n},X^{n}), & \hbox{if $h=0,$}
                                             \end{array}
                                           \right.
            \end{equation*}
            and
              \begin{equation*}
            \delta_{i}(t_{n+\frac{1}{2}},X_{i}(t_{n+\frac{1}{2}}))=3F_{i}(t_{n+\frac{1}{2}},X^{n+\frac{1}{2}})-F_{i}(t_{n},X^{n}).
            \end{equation*}
            Thus, approximations $(\ref{e1})$ and $(\ref{e2})$ become
            \begin{equation*}
            X_{i}^{n+\frac{1}{2}}=X_{i}^{n}+\frac{h}{2}\delta_{i}(t_{n},X_{i}(t_{n}))\text{\,\,\,\,and\,\,\,\,}
            X_{i}^{n+1}=X_{i}^{n+\frac{1}{2}}+\frac{h}{2}\delta_{i}(t_{n+\frac{1}{2}},X_{i}(t_{n+\frac{1}{2}})).
            \end{equation*}
             Analogously, the two-level numerical scheme generated by $\widehat{\delta}_{i}$ should furnish the approximate solutions which satisfy
             \begin{equation}\label{97a}
            \widehat{X}_{i}^{n+\frac{1}{2}}=\widehat{X}_{i}^{n}+\frac{h}{2}\widehat{\delta}_{i}(t_{n},\widehat{X}_{i}(t_{n}))
            \end{equation}
             and
             \begin{equation}\label{97b}
            \widehat{X}_{i}^{n+1}=\widehat{X}_{i}^{n+\frac{1}{2}}+\frac{h}{2}\widehat{\delta}_{i}(t_{n+\frac{1}{2}},\widehat{X}_{i}(t_{n+\frac{1}{2}})).
            \end{equation}
            In view of equation $(\ref{81})$, we have
            \begin{equation}\label{98}
            X_{i}(t_{l+\frac{1}{2}})=X_{i}(t_{l})+\frac{h}{2}\Delta_{i}(t_{n},X_{i}(t_{l})),\text{\,\,\,\,for\,\,\,\,}l=n,n+\frac{1}{2}.
            \end{equation}
            Plugging relations $(\ref{97a})$-$(\ref{98})$, and because $\widehat{e}_{i}^{l}=X_{i}(t_{l})-\widehat{X}_{i}^{l},$ direct
            calculations result in
            \begin{equation*}
            \widehat{e}_{i}^{n+\frac{1}{2}}=\widehat{e}_{i}^{n}+\frac{h}{2}\left(\Delta_{i}(t_{n},X_{i}(t_{n}))-
             \widehat{\delta}_{i}(t_{n},\widehat{X}_{i}(t_{n}))\right)=\widehat{e}_{i}^{n}+\frac{h}{2}\left[(\widehat{\delta}_{i}(t_{n},X_{i}(t_{n}))
             -\widehat{\delta}_{i}(t_{n},\widehat{X}_{i}(t_{n})))\right.
            \end{equation*}
             \begin{equation*}
            \left.+(\Delta_{i}(t_{n},X_{i}(t_{n}))-\widehat{\delta}_{i}(t_{n},X_{i}(t_{n})))\right].
            \end{equation*}
            Taking the absolute value, it is easy to see that
             \begin{equation*}
            \left|\widehat{e}_{i}^{n+\frac{1}{2}}\right|\leq |\widehat{e}_{i}^{n}|+\frac{h}{2}\left[|\widehat{\delta}_{i}(t_{n},X_{i}(t_{n}))
             -\widehat{\delta}_{i}(t_{n},\widehat{X}_{i}(t_{n}))|+|\Delta_{i}(t_{n},X_{i}(t_{n}))-\widehat{\delta}_{i}(t_{n},X_{i}(t_{n}))|\right]\leq
            \end{equation*}
            \begin{equation}\label{99}
            |\widehat{e}_{i}^{n}|+\frac{h}{2}\left[L_{i}|X_{i}(t_{n})-\widehat{X}_{i}^{n}|+C_{4}h^{2}\right]\leq\left(1+\frac{Lh}{2}\right)
             |\widehat{e}_{i}^{n}|+\frac{C_{4}}{2}h^{3},
            \end{equation}
             where $L=\underset{1\leq i\leq9}{\max}L_{i}$. The last two estimates follows from inequalities $(\ref{94a})$ and $(\ref{96a}).$ Taking the
             maximum over $i$, estimate $(\ref{99})$ gives
             \begin{equation}\label{100a}
            \|\widehat{e}^{n+\frac{1}{2}}\|_{\infty}\leq\underset{1\leq i\leq9}{\max}\left\{\left(1+\frac{Lh}{2}\right)|
             \widehat{e}_{i}^{n}|
             +\frac{C_{4}}{2}h^{3}\right\}\leq \left(1+\frac{Lh}{2}\right)\|\widehat{e}^{n}\|_{\infty}+\frac{C_{4}}{2}h^{3}.
            \end{equation}
             In a similar way, utilizing relations $(\ref{97b})$, $(\ref{98})$, $(\ref{94a})$ and $(\ref{96b})$, one easily shows that
             \begin{equation}\label{100b}
            \|\widehat{e}^{n+1}\|_{\infty}\leq \left(1+\frac{Lh}{2}\right)\|\widehat{e}^{n+\frac{1}{2}}\|_{\infty}+\frac{C_{5}}{2}h^{3}.
            \end{equation}
            Substituting estimate $(\ref{100a})$ into $(\ref{100b})$ to obtain
            \begin{equation}\label{101}
            \|\widehat{e}^{n+1}\|_{\infty}\leq \left(1+\frac{Lh}{2}\right)^{2}\|\widehat{e}^{n}\|_{\infty}+
            \left[\left(1+\frac{Lh}{2}\right)\frac{C_{4}}{2}+\frac{C_{5}}{2}\right]h^{3}=(1+\alpha_{1h})\|\widehat{e}^{n}\|_{\infty}+
            \alpha_{2h}h^{3},
            \end{equation}
             where $\alpha_{1h}=L\left(1+\frac{Lh}{4}\right)h$ and $\alpha_{2h}=\left(1+\frac{Lh}{2}\right)\frac{C_{4}}{2}+\frac{C_{5}}{2}.$ To guarantee
            the convergence of the algorithm, the step size should satisfy $0<h\leq1.$ This restriction allows to write: $\alpha_{2h}\leq\left(1+
            \frac{L}{2}\right)\frac{C_{4}}{2}+\frac{C_{5}}{2}:=\alpha_{2}.$ This fact, together with inequality $(\ref{101})$ yield
             \begin{equation*}
            \|\widehat{e}^{n+1}\|_{\infty}\leq \left(1+\frac{Lh}{2}\right)^{2}\|\widehat{e}^{n}\|_{\infty}+\alpha_{2}h^{3}.
            \end{equation*}
            Applying Lemma $\ref{l1}$, it holds
            \begin{equation}\label{102}
            \|\widehat{e}^{n}\|_{\infty}\leq \exp(n\alpha_{1h})\|\widehat{e}^{0}\|_{\infty}+\frac{\exp(n\alpha_{1h})-1}{\alpha_{1h}}
             \alpha_{2}h^{3}.
            \end{equation}
             But it comes from the initial condition that $\widehat{e}^{0}=0.$ Using this, relation $(\ref{102})$ becomes
             \begin{equation*}
            \|\widehat{e}^{n}\|_{\infty}\leq \frac{\exp[nLh(1+\frac{Lh}{4})]-1}{L(1+\frac{L}{4})}\alpha_{2}h^{2}.
            \end{equation*}
            Since $h=\frac{T_{max}-t_{0}}{M}$ and $t_{n}=t_{0}+nh,$ then $nh=t_{n}-t_{0}\leq T_{max}$. Furthermore, $1<1+\frac{Lh}{4}\leq1+\frac{L}{4}$
            (indeed, $0<h\leq1$)). Utilizing this fact, we have
             \begin{equation}\label{103}
            \|\widehat{e}^{n}\|_{\infty}\leq \frac{\exp[LT_{max}(1+\frac{L}{4})]-1}{L}\alpha_{2}h^{2},
            \end{equation}
           which can be rewritten as
             \begin{equation}\label{104}
            |X_{i}(t_{n})-\widehat{X}_{i}^{n}|\leq \frac{\alpha}{L}\left[\exp[LT_{max}(1+\frac{L}{4})]-1\right]h^{2}<
                \frac{\alpha}{L}\left[\exp[LT_{max}(1+\frac{L}{4})]-1\right],\text{\,\,for\,\,}i=1,2,...,9,
            \end{equation}
            since $0<h\leq1.$ Setting $\upsilon=\frac{\alpha}{L}\left[\exp[LT_{max}(1+\frac{L}{4})]-1\right]>0,$ estimate $(\ref{104})$ indicates that
            $(t_{n},\widehat{X}_{i}^{n})\in D_{i}$, for $i=1,2,...,9.$\\

           Now, according to the definition of $\widehat{\delta}_{i},$ we should have $X_{i}^{n}=\widehat{X}_{i}^{n},$ $e_{i}^{n}=\widehat{e}_{i}^{n}$
           and $\delta_{i}(t_{n},X_{i}(t_{n}))=\widehat{\delta}_{i}(t_{n},\widehat{X}_{i}(t_{n})).$ This fact and estimates $(\ref{103})$-$(\ref{104})$
           provide
            \begin{equation}\label{105}
            \|X(t_{n})-X^{n}\|_{\infty}\leq \upsilon\text{\,\,\,and\,\,\,}\|e^{n}\|_{\infty}\leq \upsilon h^{2}.
            \end{equation}
            It comes from the inequality $\|u\|_{\infty}-\|v\|_{\infty} \leq \|u-v\|_{\infty}$ (for any $u,v\in\mathbb{R}^{9}$) and estimate
           $(\ref{105})$ that
             \begin{equation}\label{107}
            \|X^{n}\|_{\infty}\leq \|X(t_{n})\|_{\infty}+\upsilon.
            \end{equation}
           The analytical solution $X(\cdot)$ is bounded on the interval $[t_{0},T_{max}]$ because $(t,X(t))\in D,$ where the domain $D$ is given
           by relation
            $(\ref{D1}).$ It follows from relation $(\ref{107})$ that the approximate solution $X$ is also bounded over the interval $[t_{0},T_{max}]$.
          Hence, the proposed approach $(\ref{e1})$-$(\ref{e3})$ for solving the initial-value problem $(\ref{1})$-$(\ref{6})$ is stable.\\

          Finally, the second estimate in relation $(\ref{105})$ suggests that the new method is at least second-order convergent.
          This completes the proof of Theorem $\ref{t1}.$
         \end{proof}

         \section{Numerical experiments and Convergence rate}\label{sec5}
         In this section, we use MatLab $R2007b$ and we present a broad range of numerical evidences to illustrate and demonstrate the efficiency of the
          proposed approach applied to the mathematical model of the Covid-19 spreading. We stress that in this situation, we obtain satisfactory results,
          so our algorithm performances are not worse for multidimensional problems. We consider the particular case of the SARS-Cov-2 in Cameroon where
         the data are available in \cite{camer} and we analyze and discuss the obtained results. The other real data used in this study are taken from the
          literature \cite{20ifvr,21ifvr,camer}. More precisely, the number of people in this country is approximately $N=25.000.000$, the fraction
        at time $t,$ of people in compartment $X_{5}$ that are hospitalized is assumed equals $1$ ($p(t)=1$, because of the decision made by WHO it was
         decided to hospitalized all detected cases to reduce the transmission of the SARS-Cov-2 virus \cite{45ifvr}), $d_{X_{2}}=5.5$ days,
         $d_{g}=6$ days, $d_{X_{4}}=7.3$ days (that is, $\gamma_{X_{4}}=1/7.3$ $(day^{-1})),$ $C_{u}=0.3$ days, $T_{max}=55$ days, $t_{0}=0$
          (which corresponds to $06$ March $2020$), $h=1$ (the step size), $\underline{w}=0.00133$, $\underline{w}=0.0667,$ $\beta_{X_{2}}=0.1125$,
         $\beta_{X_{3}}=0.375$, $d_{0}=14$ (the period of convalescence, we recall that this corresponds to the time a person is still hospitalized
         after recovering from Covid-19), $\gamma_{X_{2}}=0.1818$, $\gamma_{X_{3}}=0.78895$, $C_{X_{2}}=0.3$, $k=0.13$ (efficiency of control measures),
         $\lambda=12$ (first day of the application of the control measures which corresponds to $17$ March $2020$ in Cameroon),
         $\underline{\beta}_{X_{3}}=C_{u}\times \beta_{X_{3}}=0.1125$ and $\overline{\beta}_{X_{3}}=\beta_{X_{3}}=0.375$.\\

         To prevent the transmission of the Coronavirus $2019$ in the other cities of the country, the authorities decided to suspend all fights in
          the country and to restrict the number of passengers in all public transportation and outbound trains in the cities of Yaounde, Douala,
          Bafoussam and municipalities on $17$ March $2020$ \cite{camer}. Hence, the following implementation of the control measures is considered
          in order to indicate the real situation of the measured imposed
             \begin{equation*}
           m(t)=\left\{
                  \begin{array}{ll}
                    1, & \hbox{if $t\in[05\text{\,\,March\,\,}2020,\text{\,\,}\lambda)$;} \\
                    \exp[-k(t-\lambda)], & \hbox{if $t\in[\lambda,\text{\,\,}T_{max}$),}
                  \end{array}
                \right.
             \end{equation*}
          where $\lambda$ corresponds to $17$ March $2020$. Furthermore, we assume that the fraction of detected infected people ($\theta(t)$)
         is a linear function and the disease contact rate of a person in compartment $X_{4}$ denoted $w(t)$ (without taking into account the
           control measures) in this territory is a continuous function and are given by
         \begin{equation*}
          \theta(t)=\overline{w}+\frac{1-\overline{w}}{T_{max}}t\text{\,\,\,and\,\,\,}
         \beta_{X_{4}}(\theta(t))=\underline{\beta}_{X_{3}}+(\beta_{X_{3}}-\underline{\beta}_{X_{3}})\frac{1-\theta(t)}{1-w(t)}.
         \end{equation*}
          We recall that the initial data used in our experiments are taken in \cite{camer}. Specifically, as on $06$ March $2020$, there were two
         cases in which the first one was imported on $24$ February $2020$ and the other case was occurring by contact with the first case
         reported. By devoting resources, on $08$ March $2020,$ $108$ of the $176$ identified individuals of local transmission were traced back
         to their presumed exposure, either to a known case or to a location linked to spread \cite{camer}. From this observation, we set
         $X_{1}(t_{0})=N-176,$ $X_{2}(t_{0})=172,$ $X_{3}(t_{0})=X_{4}(t_{0})=1,$ $X_{5}(t_{0})=2,$ $X_{j}(t_{0})=0,$ for $j=6,7,8,9.$ Using the
         above tools together with the data provided in \cite{camer}, we perform a wide set of numerical experiments to demonstrate the robustness
         of the new approach. More specifically:\\

         In Table 1 we show the evolution of the model cumulative number of cases in Cameroon from $06$ March $2020$ to $30$ April
          $2020$ subjects to various values of the fraction of detected infected people that are documented ($\theta$), whereas Figure $\ref{fig1}$
         (\textbf{Figure 1}) deals with the evolution of the predicted number of cases in the same country during the considered time interval
         (here, $T_{max}=55$ days), taking into account different values of $\theta$. The model cumulative number of reported deaths and computed
         ones are presented in both Table 2 and Figure $\ref{fig1}$ (\textbf{Figure 2}). Also in this case, we consider different values of $\theta$.
         Table 3 and Figure $\ref{fig1}$ (\textbf{Figure 3}) indicate the model cumulative number of hospitalized people using various values of $\theta$
         while both Table 4 and Figure $\ref{fig1}$ (\textbf{Figure 4}) suggest the number of infected individuals, who are not expected to be
         detected yet, may infect other persons and start to developing clinical signs. In both Table 5 and Figure $\ref{fig2}$ (\textbf{Figure 5}),
         we present the expected number of people that will recover, but can still infect other persons. In Table 6 and Figure $\ref{fig2}$
         (\textbf{Figure 6}), we show the evolution of the cumulative number of people who recovered after being previously infected but were
         undetected and documented by the government and are no longer infectious. The model cumulative number of individuals exposed to the Covid-19
          disease is indicated in both Table 7 and Figure $\ref{fig2}$ (\textbf{Figure 7}). Both Table 8 and Figure $\ref{fig2}$ (\textbf{Figure 8})
         present the evolution of the number of persons infected by contact with people in compartment $X_{2}$, while Table 9 together with
         Figure $\ref{fig3}$ (\textbf{Figure 9}) show the model cumulative number of people infected by contact with individuals in compartment $X_{4}.$
         In Table 10 and Figure $\ref{fig3}$ (\textbf{Figure 10}), we present the number of persons infected by contact with people in compartment
         $X_{10}=X_{5}+X_{6}.$ Each analysis described above considers various values of $\theta.$ Finally, we draw in Table 11 and sketch in Figure
         $\ref{fig3}$ (\textbf{Figure 11}) the effective reproduction number of SARS-Cov-2 for Cameroon.\\

         Now, focusing on the values taken by $\theta,$ we observe from the tables and figures that for a greater value of $\theta$ (considered in
          this analysis) the predicted results reproduce quite accurately the evolution of the number of cases, number of deaths and number of
         hospitalized people. In particular, for $\theta=0.3212$ the predicted values on $30$ April $2020$ of cases, deaths and people in
         hospital are approximately $2236.0$, $111.2786$ and $47.4915$, respectively, while these values become overestimated: $3260.6$ cases,
         $166.1457$ deaths and $90.1832$ hospitalized people when $\theta=\overline{w}=0.0667$ (the smallest value taken by $\theta$ in
          $[\overline{w},1]$). Furthermore, for various values of the parameter $\theta$ we observe that the maximum number of undetected cases: $116.0502$ people
         (for $\theta=0.0667$), $85.1727$ people (for $\theta=0.1515$), $74.3269$ people (for $\theta=0.3212$) and $56.5865$ people
         (for $\theta=0.7455$) estimated by the new approach represent: $5.5\%$ (for $\theta=0.0667$), $4.2257\%$ (for $\theta=0.1515$), $4.0439\%$
         (for $\theta=0.3212$) and $3.6912\%$ (for $\theta=0.7455$) of the number of total cases obtained on $30$ March $2020$. Interestingly,
          the results provided by our method suggest that, despite the relative control of Covid-19 pandemic in Cameroon,
         they may still exist an undetected source of infected persons that could cause the increase of the disease in a near future, if the
         implementation of the control measures are significantly relaxed like the government decided at the beginning of May $2020$. In addition,
         Figure $\ref{fig1}$ (\textbf{Figure 3}) indicates that the peak of the persons hospitalized in this country at the same time should be
         reached around
         $05$ April $2020$, with approximately $800$ hospitalized patients. This number is associated to the smallest value of $\theta=0.0667$ considered
         in this study. However, the obtained results slightly overestimate the observed values by around $693.$ Focusing on the recovered people, the
         considered technique suggests a maximum number of $3037.7$ people, $2300.8$ people, $2073.5$ people, and $1686.0$ people corresponding to
         various values of $\theta:$ $0.0667$, $0.1515$, $0.3212$ and $0.7455$, respectively. For $\theta=0.1515,$ the final number of $241.0999$
         people is very close to the real observation (around $244$ hospitalized individuals on $15$ April $2020$). This shows that the proposed
         method is a reasonable decision tool to estimate the number of beds in hospital during a pandemic or an epidemic. Also, it is worth mentioning
         that our approach is able to detect early a reasonable expected date of this peak.\\

         Finally, we observe from both Table 11 and Figure $\ref{fig3}$ (\textbf{Figure 11}) that the value of the effective reproduction number
        ($R_{e}$) decreases since the application of the control measures and it becomes less than $1$ after the peak is attained ($30$ March $2020$).

         \textbf{Tables 1.} Model cumulative number of cases with various values of $\theta$
           $$\begin{tabular}{|c|c|c|c|c|c|}
            \hline
               & 06 March 2020 & 17 March 2020 & 30 March 2020 & 15 April 2020 & 30 April 2020 \\
            \hline
            NC(real data)& 2 & 15       & 142    & 848  & 2014  \\
            \hline
            CM($0.0667$) & 2 & 696.3508 & 2149.8 & 3245.3 & 3322.0\\
            \hline
            CM($0.1515$) & 2 & 681.7603 & 2048.7 & 2467.4 & 2481.5 \\
            \hline
            CM($0.3212$) & 2 & 647.5659 & 1861.4 & 2223.8 & 2236.0\\
            \hline
            CM($0.7455$) & 2 & 583.5342 & 1538.2 & 1808.4 & 1817.5\\
            \hline
          \end{tabular}$$
           \text{\,}\\
         \textbf{Tables 2.} The evolution of number of deaths with different values of $\theta$
           $$\begin{tabular}{|c|c|c|c|c|c|}
            \hline
               & 06 March 2020 & 17 March 2020 & 30 March 2020 & 15 April 2020 & 30 April 2020 \\
            \hline
            ND(real data)   & 0 & 1       & 6       & 14       &  61\\
            \hline
            $X_{9}(0.0667$) & 0 & 17.8300 & 79.6390 & 141.2279 & 166.1457 \\
            \hline
            $X_{9}(0.1515$) & 0 & 17.5511 & 73.8069 & 108.6332 & 121.8038\\
            \hline
            $X_{9}(0.3212$) & 0 & 16.8994 & 68.3353 & 99.5074  & 111.2786 \\
            \hline
            $X_{9}(0.7455$) & 0 & 15.6638 & 58.6885 & 83.6826  & 93.0934 \\
            \hline
          \end{tabular}$$

         \textbf{Tables 3.} The cumulative number of hospitalized people
           $$\begin{tabular}{|c|c|c|c|c|c|}
            \hline
               & 06 March 2020 & 17 March 2020 & 30 March 2020 & 15 April 2020 & 30 April 2020 \\
            \hline
            NH (real obs.) & 2 & 8        & 61       & 244      & 328\\
            \hline
            Host($0.0667$) & 2 & 335.5282 & 668.8886 & 476.0395 & 86.9377 \\
            \hline
            Host($0.1515$) & 2 & 325.1239 & 673.6724 & 241.0999 & 44.0635\\
            \hline
            Host($0.3212$) & 2 & 300.8178 & 583.1121 & 209.3954 & 47.4915 \\
            \hline
            Host($0.7455$) & 2 & 255.0717 & 432.0405 & 158.2955 & 54.7413 \\
            \hline
          \end{tabular}$$
            These numbers should help to expect the number of beds in hospitals during an epidemic.
            \text{\,}\\
           \text{\,}\\
         \textbf{Tables 4.} Cumulative number of infected individuals who are not expected to be detected yet, but start developing clinical signs
           $$\begin{tabular}{|c|c|c|c|c|c|}
            \hline
               & 06 March 2020 & 17 March 2020 & 30 March 2020 & 15 April 2020 & 30 April 2020 \\
            \hline
            $X_{3}(0.0667$) & 1 & 70.3135 & 116.0502 & 15.4822 & 0.2363 \\
            \hline
            $X_{3}(0.1515$) & 1 & 67.8561 & 85.1727 & 2.6167 & 0.0608\\
            \hline
            $X_{3}(0.3212$) & 1 & 62.0955 & 74.3269 & 2.2831 & 0.0530 \\
            \hline
            $X_{3}(0.7455$) & 1 & 51.8569 & 56.5865 & 1.7375 & 0.0404 \\
            \hline
          \end{tabular}$$
           \text{\,}\\
           \text{\,}\\
         \textbf{Tables 5.} Model cumulative number of infected people that will recover but can still infect other persons
           $$\begin{tabular}{|c|c|c|c|c|c|}
            \hline
               & 06 March 2020 & 17 March 2020 & 30 March 2020 & 15 April 2020 & 30 April 2020 \\
            \hline
            $X_{5}(0.0667$) & 2 & 307.1123 & 608.9473 & 486.4522 & 106.7015 \\
            \hline
            $X_{5}(0.1515$) & 2 & 297.3910 & 626.9163 & 268.6669 & 52.8153\\
            \hline
            $X_{5}(0.3212$) & 2 & 274.6855 & 547.2620 & 232.7084 & 45.7222 \\
            \hline
            $X_{5}(0.7455$) & 2 & 231.8949 & 413.7484 & 173.1906 & 33.9913 \\
            \hline
          \end{tabular}$$
           \text{\,}\\
         \textbf{Tables 6.} Evolution of the number of undetected infected people who recovered but are no longer infectious
           $$\begin{tabular}{|c|c|c|c|c|c|}
            \hline
               & 06 March 2020 & 17 March 2020 & 30 March 2020 & 15 April 2020 & 30 April 2020 \\
            \hline
            $X_{5}(0.0667$) & 0 & 342.9927 & 1402.0 & 2582.4 & 3037.7 \\
            \hline
            $X_{5}(0.1515$) & 0 & 339.0853 & 1299.3 & 2070.9 & 2300.8\\
            \hline
            $X_{5}(0.3212$) & 0 & 329.8487 & 1202.1 & 1874.4 & 2073.5 \\
            \hline
            $X_{5}(0.7455$) & 0 & 312.7987 & 1030.5 & 1537.8 & 1686.0 \\
            \hline
          \end{tabular}$$
           \text{\,}\\
           \textbf{Tables 7.} Cumulative number of people exposed to Covid-19 disease
           $$\begin{tabular}{|c|c|c|c|c|c|}
            \hline
               & 06 March 2020 & 17 March 2020 & 30 March 2020 & 15 April 2020 & 30 April 2020 \\
            \hline
            $X_{2}(0.0667$) & 172 & 317.5650 & 495.4672 & 52.1206 & 0.7125 \\
            \hline
            $X_{2}(0.1515$) & 172 & 306.0850 & 289.9048 & 8.6553 & 0.2011\\
            \hline
            $X_{2}(0.3212$) & 172 & 279.0561 & 252.9417 & 7.5517 & 0.1754 \\
            \hline
            $X_{2}(0.7455$) & 172 & 231.3889 & 192.4997 & 5.7472 & 0.1335 \\
            \hline
          \end{tabular}$$
           \text{\,}\\
           \textbf{Tables 8.} Model cumulative number of people infected by contact with persons in compartment $X_{2}$
           $$\begin{tabular}{|c|c|c|c|c|c|}
            \hline
               & 06 March 2020 & 17 March 2020 & 30 March 2020 & 15 April 2020 & 30 April 2020 \\
            \hline
            $X_{2}(0.0667$) & 19.3499 & 35.7245 & 37.7351 & 0.6432 & 0.0013 \\
            \hline
            $X_{2}(0.1515$) & 19.3499 & 34.4331 & 0.0000 & 0.0000 & 0.0000\\
            \hline
            $X_{2}(0.3212$) & 19.3499 & 31.3926 & 0.0000 & 0.0000 & 0.0000 \\
            \hline
            $X_{2}(0.7455$) & 19.3499 & 26.0303 & 0.0000 & 0.0000 & 0.0000 \\
            \hline
          \end{tabular}$$
           \text{\,}\\
           \textbf{Tables 9.} Model cumulative number of people infected by contact with persons in compartment $X_{4}$
           $$\begin{tabular}{|c|c|c|c|c|c|}
            \hline
               & 06 March 2020 & 17 March 2020 & 30 March 2020 & 15 April 2020 & 30 April 2020 \\
            \hline
            $X_{4}(0.0667$) & 0.3750 & -9.1174 & -12.2623 & -1.5011 & -0.0477 \\
            \hline
            $X_{4}(0.1515$) & 0.3641 & -8.0160 & 0.0000      & 0.0000     & 0.0000\\
            \hline
            $X_{4}(0.3212$) & 0.3368 & -5.6330 & 0.0000      & 0.0000     & 0.0000 \\
            \hline
            $X_{4}(0.7455$) & 0.2686 & -1.5867 & 0.0000      & 0.0000     & 0.0000 \\
            \hline
          \end{tabular}$$
           \text{\,}\\
           \textbf{Tables 10.} Model cumulative number of people infected by contact with persons in compartment $X_{10}=X_{5}+X_{6}$
           $$\begin{tabular}{|c|c|c|c|c|c|}
            \hline
               & 06 March 2020 & 17 March 2020 & 30 March 2020 & 15 April 2020 & 30 April 2020 \\
            \hline
            $X_{10}(0.0667$) & 0.0044 & 0.7379 & 0.9748   & 0.1191 & 0.0038 \\
            \hline
            $X_{10}(0.1515$) & 0.0058 & 0.9354 & 000      & 000     & 000\\
            \hline
            $X_{10}(0.3212$) & 0.0069 & 1.0414 & 000      & 000     & 000 \\
            \hline
            $X_{10}(0.7455$) & 0.0076 & 0.9745 & 000      & 000     & 000 \\
            \hline
          \end{tabular}$$
           \text{\,}\\
           \textbf{Tables 11.} Effective reproduction number of SARS-Cov-2 for Cameroon
           $$\begin{tabular}{|c|c|c|c|c|c|}
            \hline
               & 06 March 2020 & 17 March 2020 & 30 March 2020 & 15 April 2020 & 30 April 2020 \\
            \hline
            Re$(0.0667$) & 1.1111 & 1.1111 & 0.7546   & 0.1229 & 0.0175  \\
            \hline
            Re$(0.1515$) & 1.1164 & 1.1163 & 0.0000    & 0.0000   & 0.0000\\
            \hline
            Re$(0.3212$) & 1.1209 & 1.1208 & 0.0000   & 0.0000  & 0.0000 \\
            \hline
            Re$(0.7455$) & 1.1236 & 1.1236 & 0.0000   &0.0000   & 0.0000 \\
            \hline
          \end{tabular}$$

         \section{General conclusion and future works}\label{sec6}
         We have developed an efficient two-level explicit method for estimating the propagation of Covid-19 disease with undetected infectious cases. The
         analysis has shown that the new algorithm is stable, at least second-order accuracy and can serve as a fast and robust tool for integrating of
         general systems of ordinary differential equations. Numerical results based on the case of Cameroon reproduce quite accurately the evolution of
         the number of cases (detected or undetected), number of deaths, number of people in hospitals, number of infected detected persons who recovered
         and number of infected undetected individuals who recovered by natural immunity from $06$ March $2020$ to $30$ April $2020.$ The approach
         presented in this work can help to estimate the number of beds in hospitals during a pandemic. Furthermore, the proposed technique can be
         considered as a fundamental tool for detecting early a reasonable expected date of the peak during an epidemic. Our future works will consider
         the numerical solution of a more complex system of ordinary differential equations using the new two-level explicit approach.\\

         \textbf{Acknowledgment.} This work has been partially supported by the deanship of scientific research of Al-Imam Muhammad Ibn Saud Islamic
       University (IMSIU) under the Grant No. 331203.

            \newpage
         \begin{figure}
         \begin{center}
          Model cumulative number of cases, deaths, hospitalized people and infected persons who are not detected yet.
          \begin{tabular}{c c}
         \psfig{file=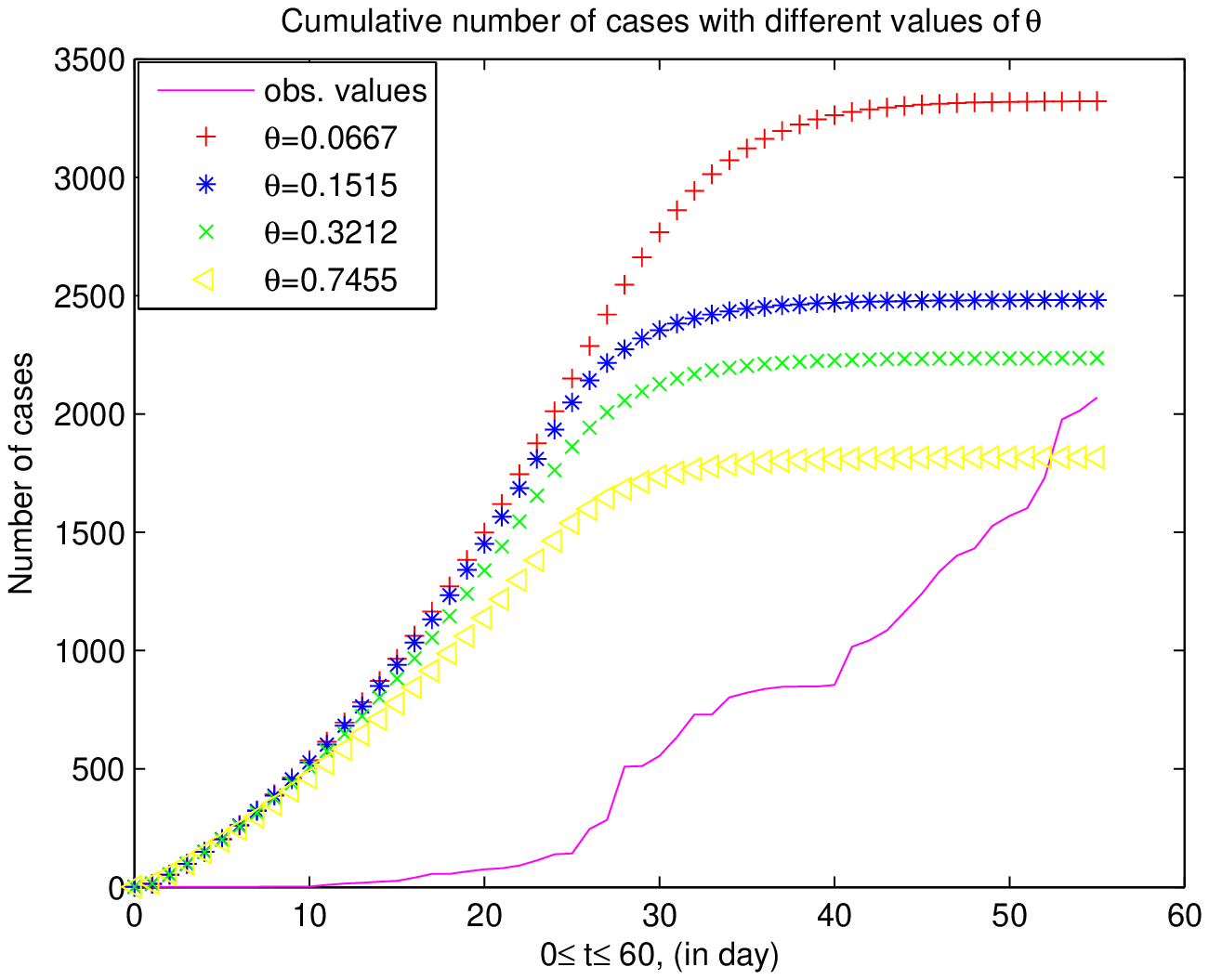,width=8cm}  & \psfig{file=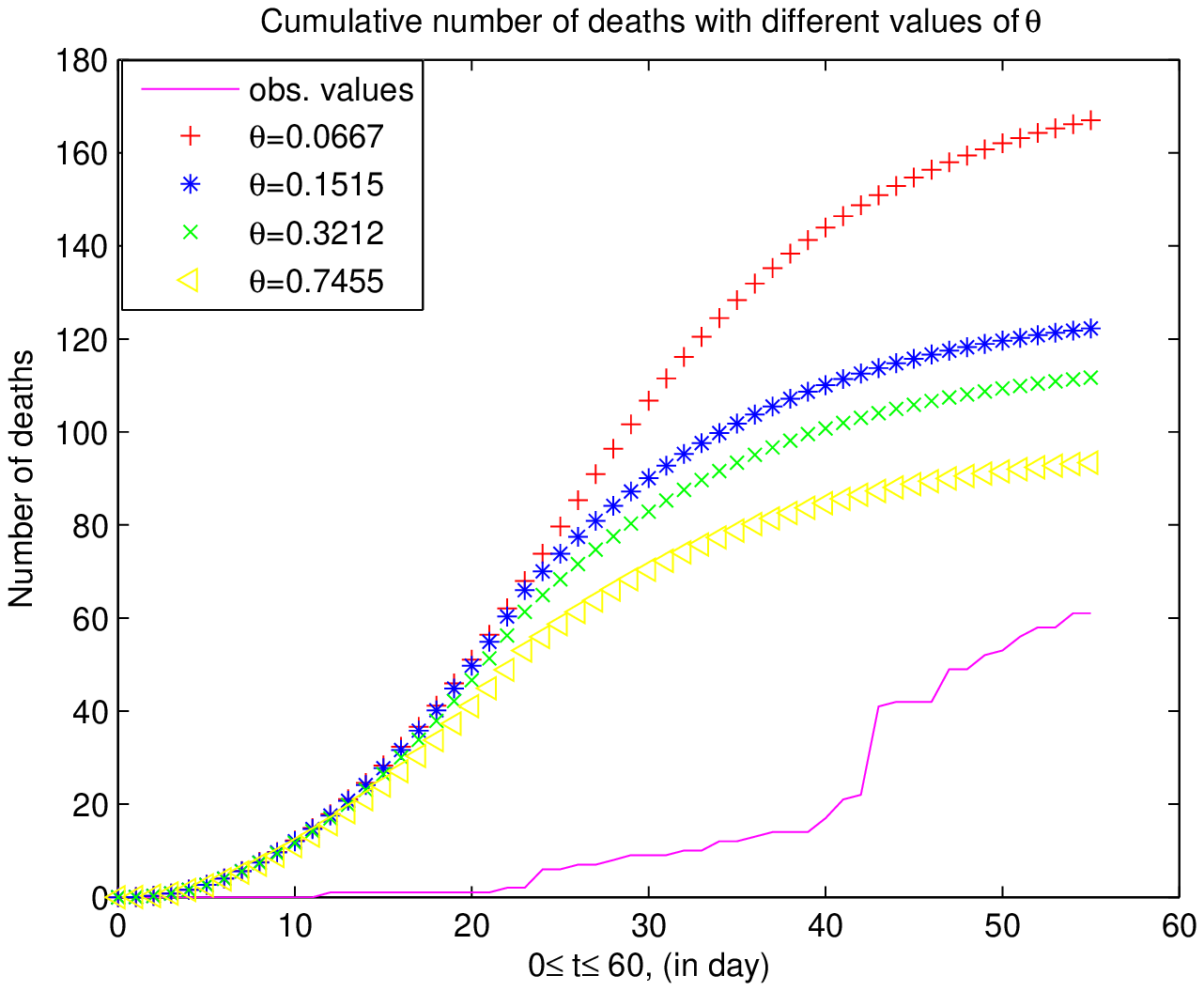,width=8cm}\\
                    \textbf{Figure 1}  & \textbf{Figure 2}\\
          \psfig{file=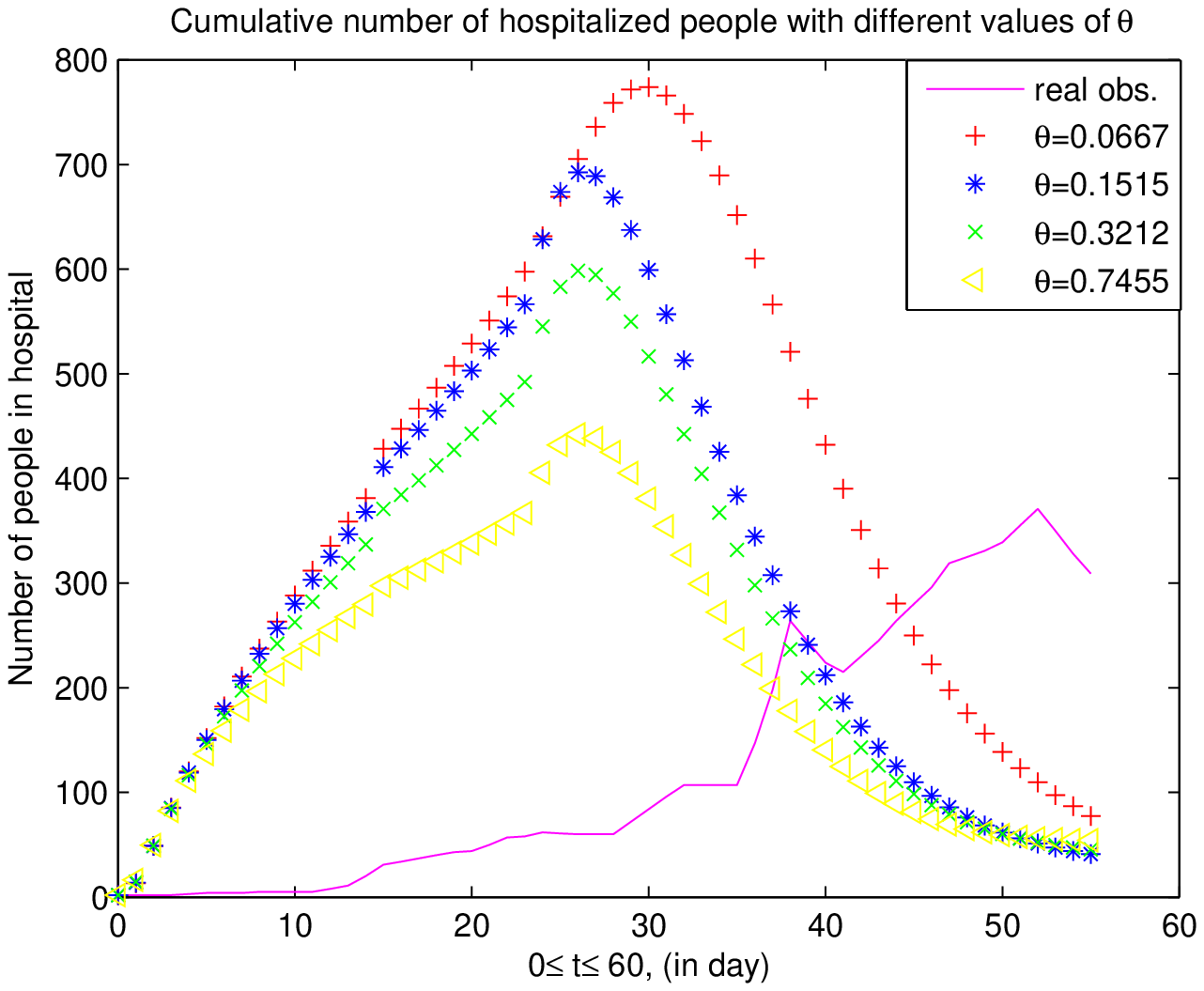,width=8cm}  & \psfig{file=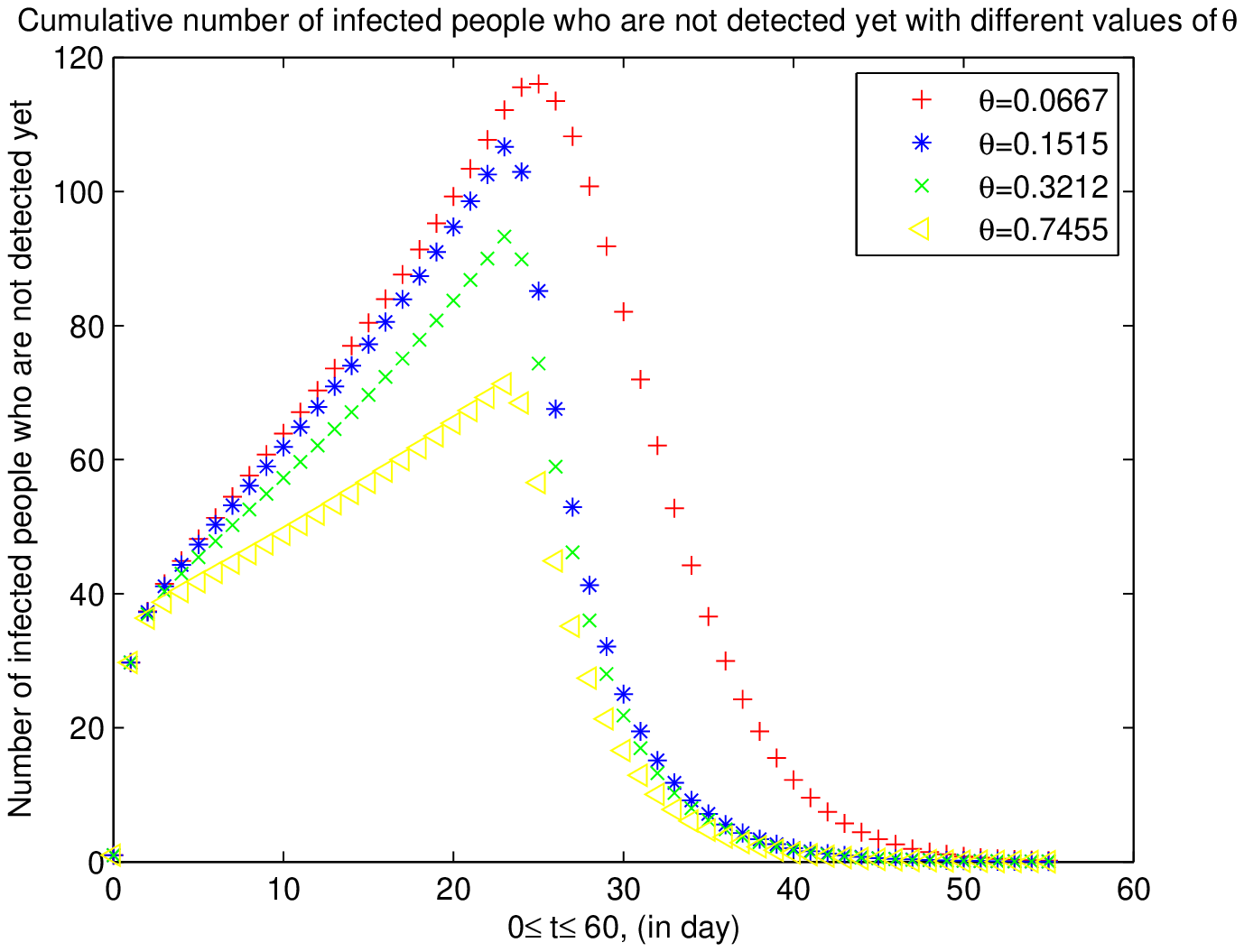,width=8cm}\\
           \textbf{Figure 3}  & \textbf{Figure 4}\\
         \end{tabular}
        \end{center}
         \caption{Number of cases, deaths, hospitalized people and undetected infected individual}
          \label{fig1}
          \end{figure}
         \begin{figure}
         \begin{center}
         Cumulative number of recovered infected, undetected infected who recovered, people exposed to Covid-19, infected by contact with $X_{2}$.
          \begin{tabular}{c c}
         \psfig{file=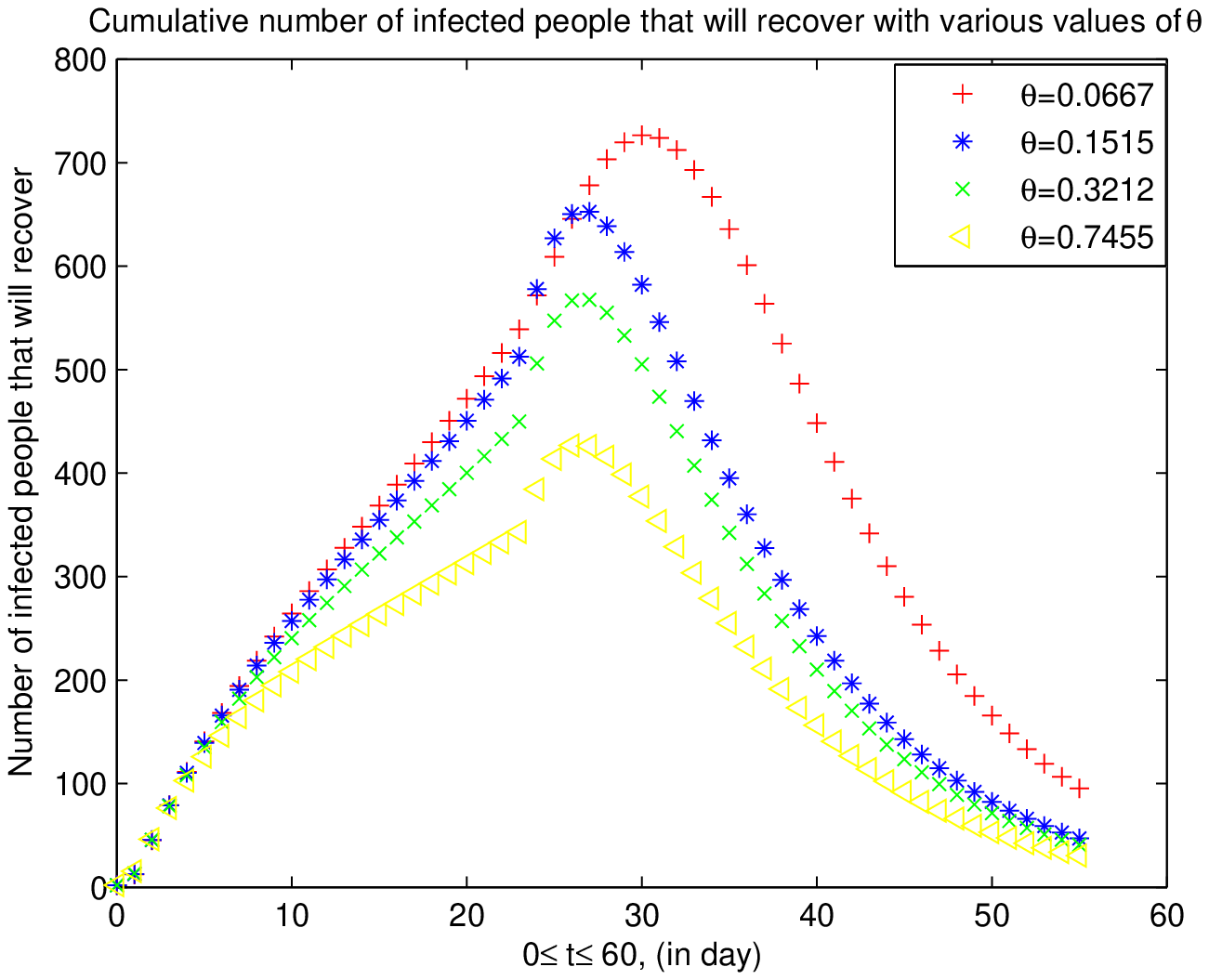,width=8cm}  & \psfig{file=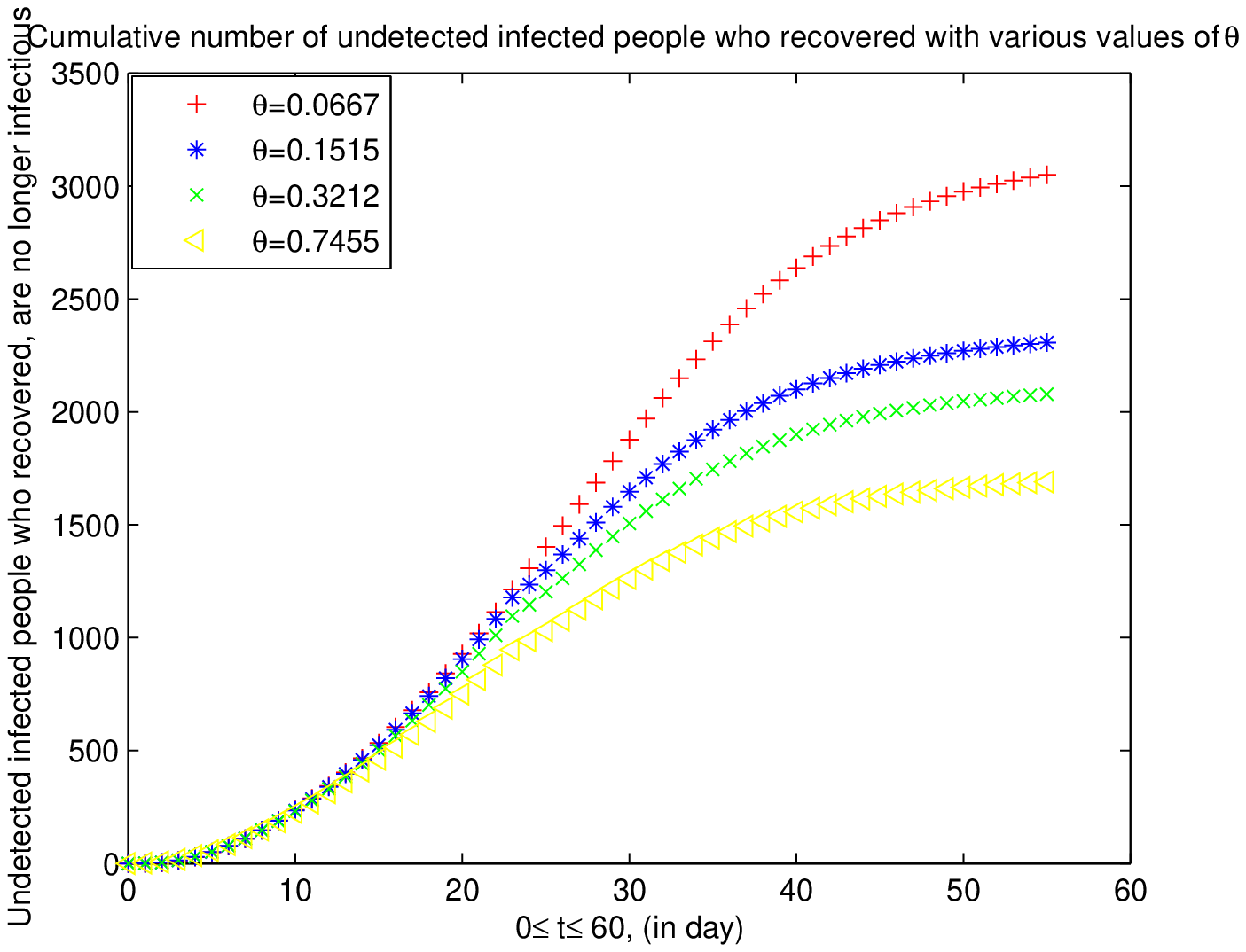,width=8cm}\\
                                             \textbf{Figure 5}  & \textbf{Figure 6}\\
          \psfig{file=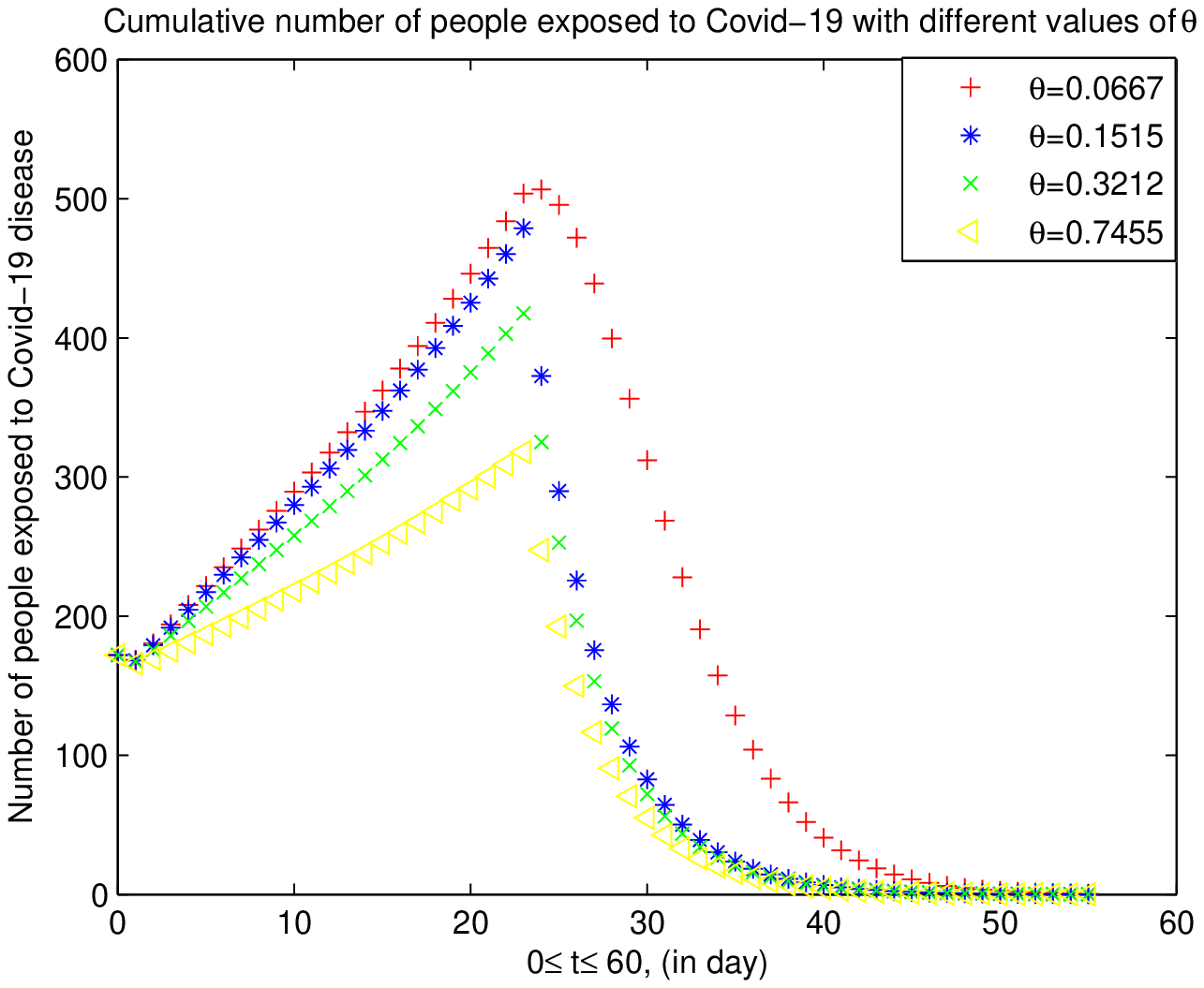,width=8cm}  & \psfig{file=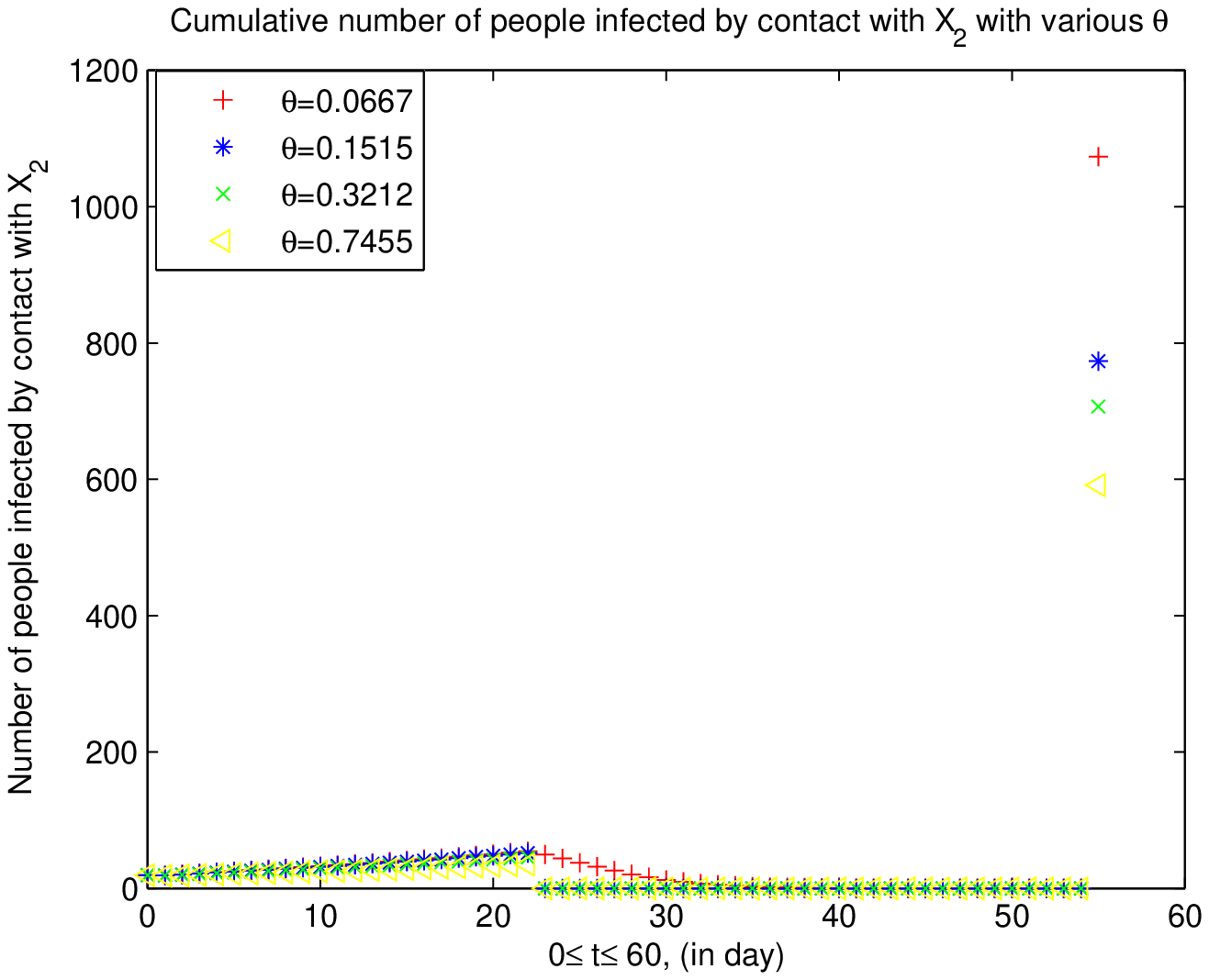,width=8cm}\\
                                     \textbf{Figure 7}  & \textbf{Figure 8}\\
         \end{tabular}
        \end{center}
         \caption{Number of recovered infected, recovered undetected infected, exposed to Covid-19, infected by contact with $X_{2}$}
          \label{fig2}
          \end{figure}
         \begin{figure}
         \begin{center}
          Model cumulative number of people infected by contact with $X_{4}$, $X_{10}$ and effective reproduction number.
          \begin{tabular}{c c}
         \psfig{file=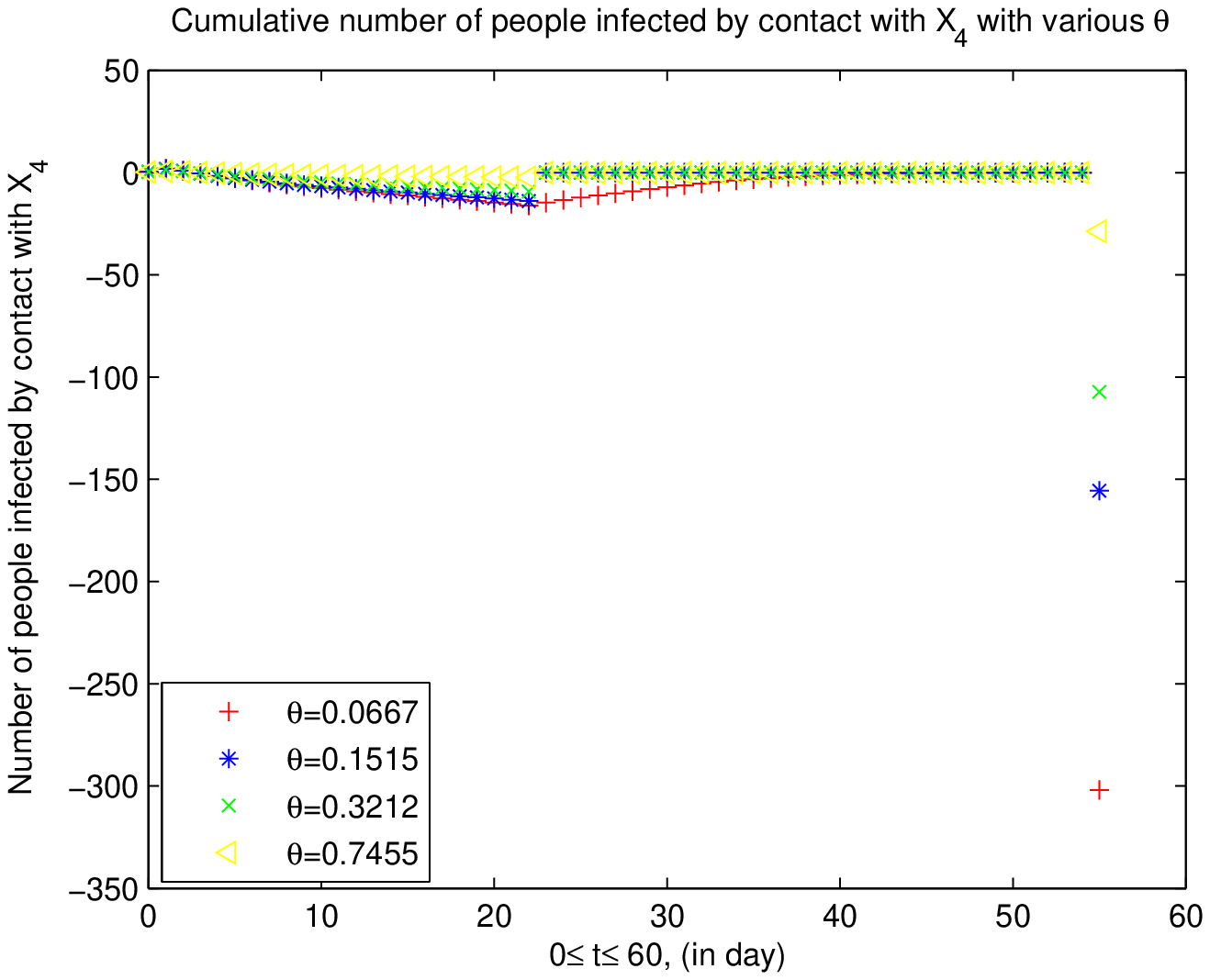,width=8cm}  & \psfig{file=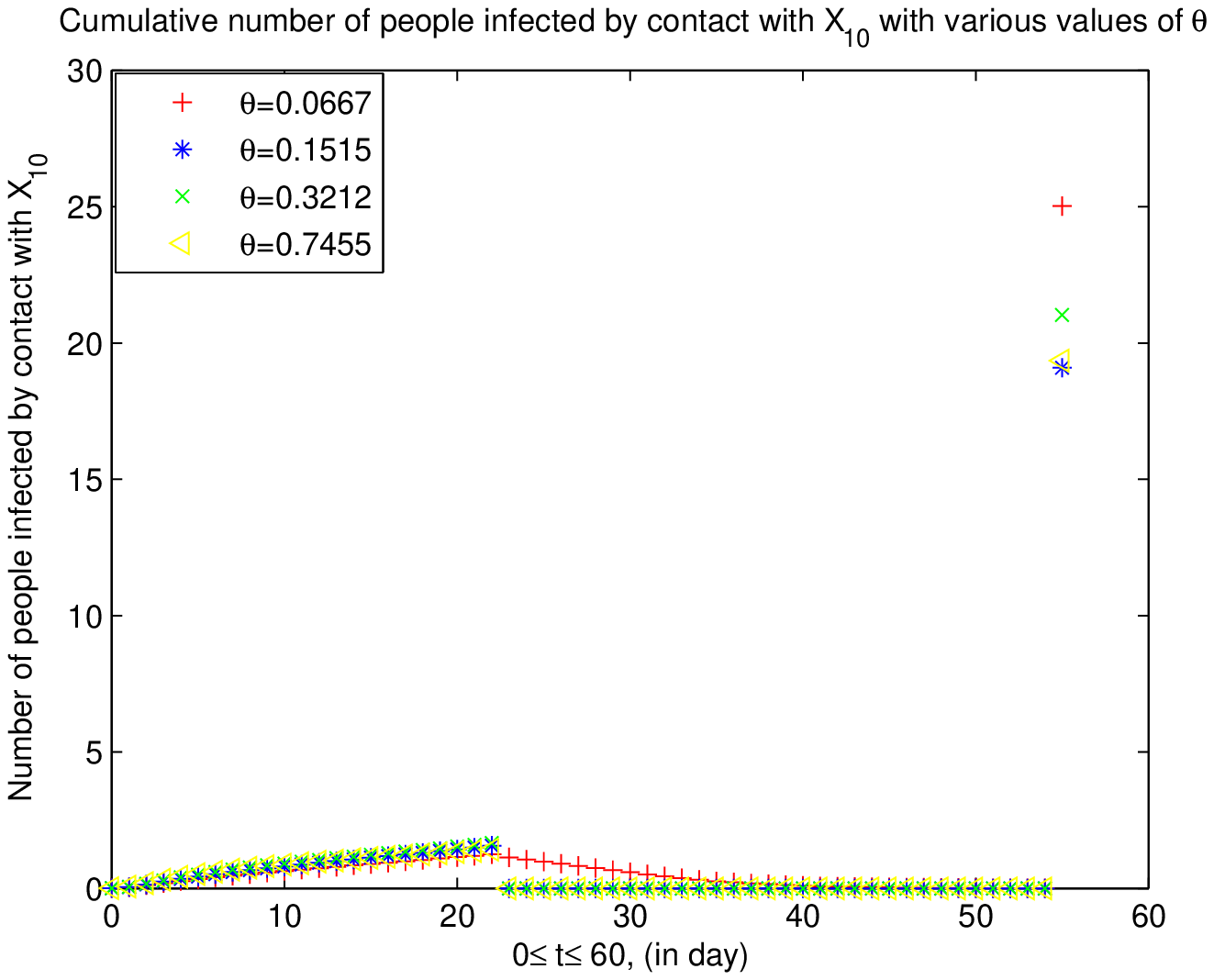,width=8cm}\\
                          \textbf{Figure 9}  & \textbf{Figure 10}\\
          \psfig{file=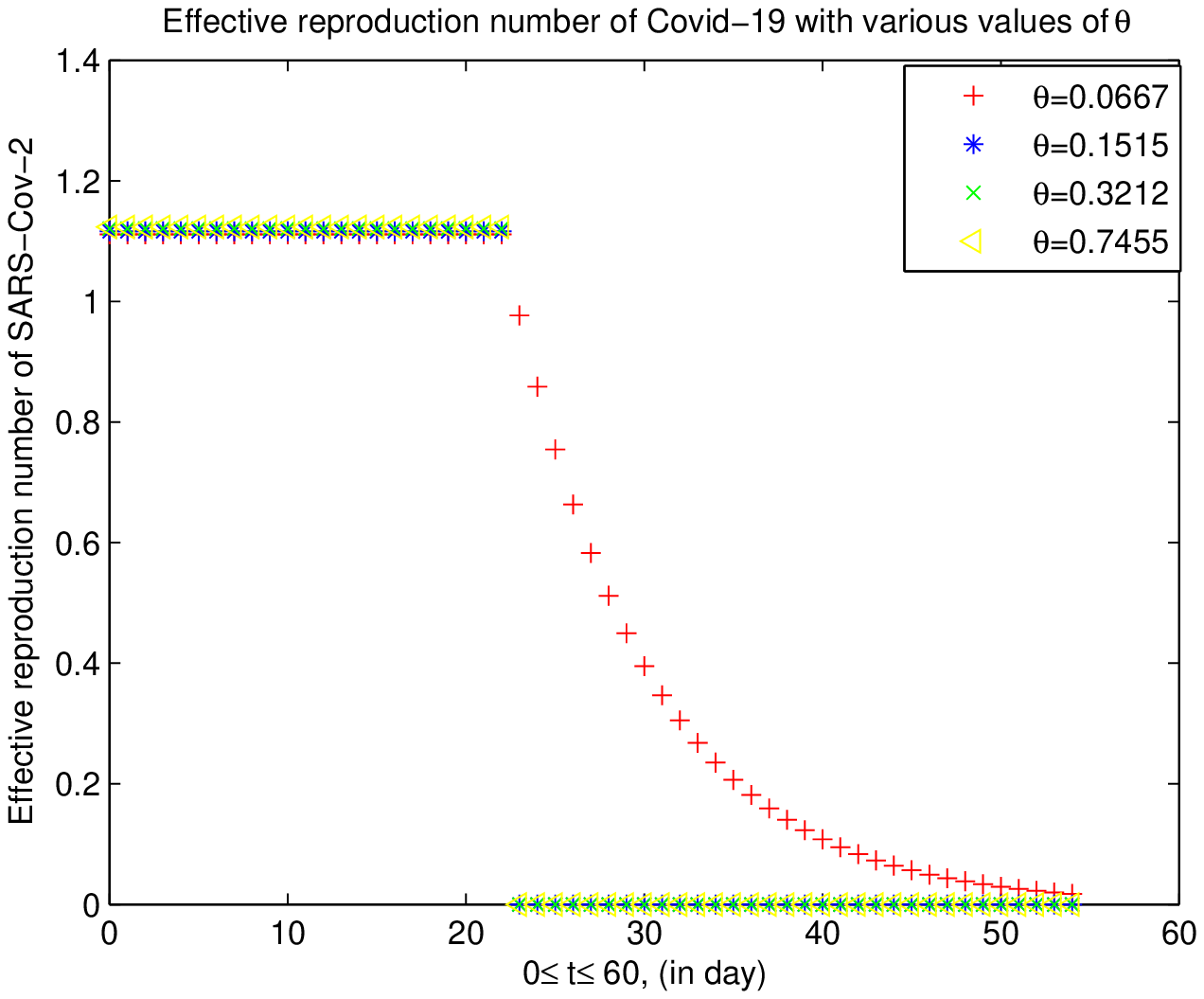,width=8cm}  & \\
                          \textbf{Figure 11}  & \\
         \end{tabular}
        \end{center}
         \caption{Number of infected by contact with $X_{4}$, $X_{10}$ and effective reproduction number $Re$}
          \label{fig3}
          \end{figure}


\begin{thebibliography}{99}

    \bibitem{3jlj}.
      J. Adam Kucharski. "Complex models may be no more reliable than simple ones if they miss key aspects of the biology. Complex models can
     create the illusion of of realism, and make it harder to spot crucial omissions", https://www.pnas.org/content/103/33/12221, April $2020$.

    \bibitem{1ifvr}.
    M. Anderson. "Population biology of infectious disease: Part 1", Nature, $280$ $(1979)$, p. $361$-$367$.

    \bibitem{13cp}
     R. M. Anderson, R. M. May. "Infection diseases of humans: dynamics and control", Oxford: Oxford universty press, $(1991)$.

     \bibitem{15cp}
     R. Beaglehole, R. Bonita, T. Kjellstrom. "Basic epidemiology. Communicable disease epidemiology", Geneva: WHO, $(1993)$, p. $97$-$105$.

     \bibitem{3ifvr}
     F. Brauer, C. Castillo-Ch\'{a}vez. "Mathematical models in population biology and epidemiology. Text in applied mathematics", Springer, $(2001)$.

     \bibitem{camer}
     Cameroon: COVID-19 Rapport de Situation No: $23$, $22$,$\cdots,4$. $05$ March-$30$ April, $(2020).$

    \bibitem{4ifvr}.
     S. Chen, J. Yang, W. Yang, T. Barnighausen. "Covis-19 control in China during mass population morments at New Year", The Nacet,
     $395(10226)$ $(2020)$, p. $764$-$766$.

     \bibitem{3sn}
    R. Codina, J. Principe, C. Munoz, J. Baiges. "Numerical modelling of chlorine concentration in water storage tanks". Int. J. Numer. Methods Fluids,
    $79$ $(2015),$ pp. $84$-$107.$

    \bibitem{tg4}
     W. Dai, R. Nassar. "Compact ADI method for solving partial differential equations". Numer. Meth. Partial Differential Eq., $18$ $(2002)$ $129$-$142.$

     \bibitem{tg8}
    S .C. R. Dennis, J. D. Hudson. "Compact $h^{4}$ finite-difference approximations to operators of Navier-Stokes type". J. Comput. Phys.,
    $85$ $(1989)$ $390$-$418.$

     \bibitem{5ifvr}.
     O. Diekmann, H. Heesterbeek, T. Britton. "Understanding infectious disease dynamics", Princeton Series in Theoretical and Computational Biology",
     Princeton University Press, $(2013)$.

     \bibitem{14sn}
    J. Douglas, T. F. Russell. "Numerical methods for convection-dominated diffusion problems based on combining the method of characteristics with
    finite element or finite difference procedures". SINUM, $19$ $(1982),$ pp. $871$-$885.$

     \bibitem{11jfvr}
    European Centre for Disease Prevention and Control. "Discharge criteria for COVID-19 cases-When is it safe to discharge Covid-19 cases
    from the hospital or end home isolation?". https://www.ecdc.europa.eu/sites/default/files/documents/{Covid}-19-Discharge-criteria.pdf, March $(2020).$

     \bibitem{tg10}
    M. M. Gupta, R. P. Manohar, J. W. Stephenson. "A single cell high order scheme for the convection-diffusion equation with variable coefficients",
    Int. J. Numer. Meth. Fluids, $4$ $(1984)$ $641$-$651.$

    \bibitem{gv1983}
    V. Guvanasen, R. E. Volker. "Numerical solutions for solute transport in unconfined aquifers". Int. J. Num. Meth. Fluids, $3$ $(1983)$ $103$-$123.$

     \bibitem{14cp}
     M. E. Halloran. "Concepts of infectious disease epidemiology. In: Rothman KJ, Greedland S. eds", Modern epidemiology, 2nd edn. Philadelphia,
     PA: Lippincott-Raven, $(2003)$, p. $519$-$554$.

    \bibitem{he1962}
     P. Henrici. "Discrete variable methods in ordinary differential equations". New York, John Wiley $(1962)$.

     \bibitem{16ifvr}.
     B. Ivorra, A. M. Ramos, D. Ngom. "Be-CoviDis: A mathematical model to predict the risk of human diseases spread between-countries. Validation and
     application to the $2014$ Ebola Virus Disease epidemic", Bulletin of Mathematical Biology, $77(9),$ $(2015),$ p. $1668$-$1704$.

     \bibitem{ifvr}.
     B. Ivorra, M. R. Ferr\'{a}ndez, M. Vela-P\'{e}rez, A. M. Ramos. "Mathematical modeling of the spread of the coronavirus disease 2019 (Covid-19)
     taking into account the undetected infectious. The case of China". Instituto de Matematica Interdisciplinar, April $(2020)$. DOI link:
      http://www.doi.org/10.13140/RG.2.2.21543.29604

      \bibitem{11sn}
    V. John, J. Novo. "On (essentially) non-oscillatory discretizations of evolutionary convection-diffusion equations". J. Comput. Phy.,
    $231$ $(2012),$ $1570$-$1585.$

    \bibitem{tg13}
    S. Karaa, J. Zhang. "Higher order ADI method for solving unsteady convection-diffusion problems". J. Comput. Phys., $198$ $(2004)$ $1$-$9.$

     \bibitem{13tls}
    P. D. Lax, B. Wendroff. "Systems of conservation laws". Comm. Pure $\&$ Appl. Math. $13$ $(1960)$ $217$-$237.$

    \bibitem{19ifvr}.
     P. Lekone, B. Finkenst\"{a}dt. "Statistical inference in a stochastic epidemic seir model with control intervention: Ebola as a case study",
     Biometric, $62(4)$ $(2006)$, p. $1170$-$1177$.

     \bibitem{20ifvr}.
     R. Li, S. Pei, B. Chen, Y. Song, T. Zhang, W. Yang, J. Shaman. "Substantial undocumented infection facilitates the rapid dissemination of
     novel coronavirus (SARS-Cov2)", Science $(2020)$.

     \bibitem{21ifvr}.
     T. Liu, J. Hu, M. Kang, L. Lin, H. Zhong, J. Xiao, and et al. "Transmission dynamics of $2019$ novel coronavirus ($2019$-ncov)", $bioRxiv$ $(2020)$.

     \bibitem{23ifvr}.
     W. Luo, M. Majumder, D. Liu. "The role of absolute humidity on transmission rates of the COVID-19 outbreak", $MedRxiv$ $(2020)$.

     \bibitem{18tls}
    R. W. MacCormack. "An efficient numerical method for solving the time-dependent compressible Navier-Stokes equations at high Reynolds numbers".
    NASA TM $(1976)$ $73$-$129.$

    \bibitem{tg15}
    R. J. MacKinnon, R. W. Johnson. "Differential equation based representation of truncation errors for accurate numerical simulation",
    Int. J. Numer. Meth. Fluids, $13$ $(1991)$ $739$-$757.$

    \bibitem{30mt}
    C. Man, C. W. Tsai. "A high order predictor-corrector scheme for two-dimensional advection-diffusion equation". Int. J. Numer. Methods Fluids,
    $56(4),$ $(2008)$ $401$-$418.$

     \bibitem{24ifvr}.
     B. Mart\'{I}nez-L\'{o}pez, B. Ivorra, A. M. Ramos, J. M. S\'{a}nchez-Vizcaino. "A novel spatial and stochastic model to evaluate the within- and
     between-farm transmission of classical swine fever virus. I. General concepts and description of the model", Veterinary Microbiology,
      $147(3-4)$ $(2011)$, p. $300$-$309$.

    \bibitem{25tls}
    F. T. Namio, E. Ngondiep, R. Ntchantcho, J. C. Ntonga. "Mathematical models of complete shallow water equations with source terms, stability
    analysis of Lax-Wendroff scheme". J. Theor. Comput. Sci., Vol. $2(132)$ $(2015).$

    \bibitem{en6}.
    E. Ngondiep. "An efficient three-level explicit time-split scheme for solving two-dimensional unsteady nonlinear coupled Burgers's equations",
     Int. J. Numer. Methods Fluids. November $(2019)$, p. $1$-$19$, $19$ pages.

    \bibitem{en1}
    E. Ngondiep. "Stability analysis of MacCormack rapid solver method for evolutionary Stokes-Darcy problem", J. Comput. Appl. Math. $345(2019)$,
    $269$-$285$, $17$ pages.

    \bibitem{en2}
    E. Ngondiep. "Long Time Stability and Convergence Rate of MacCormack Rapid Solver Method for Nonstationary Stokes-Darcy Problem",
    Comput. Math. Appl., Vol $75$, $(2018)$, $3663$-$3684,$ $22$ pages.

     \bibitem{en5}
     E. Ngondiep. "An efficient three-level explicit time-split method for solving $2$D heat conduction equations", submitted.

    \bibitem{en4}
    E. Ngondiep. "Long time unconditional stability of a two-level hybrid method for nonstationary incompressible Navier-Stokes equations",
    J. Comput. Appl. Math. $345(2019)$, $501$-$514$, $14$ pages.

    \bibitem{en}
     E. Ngondiep. "Asymptotic growth of the spectral radii of collocation matrices approximating elliptic boundary problems", Int. J. Appl. Math. Comput.,
     $4(2012)$, $199$-$219,$ $20$ pages.

    \bibitem{en3}
    E. Ngondiep. " Error estimates of MacCormack rapid solver method for nonstationary incompressible Navier-Stokes equations", preprint available
    from http://arXiv.org/abs/$1903.10857,$ $(2019)$ $26$ pages.

    \bibitem{en10}
    E. Ngondiep. "A novel three-level time-split MacCormack scheme for two-dimensional evolutionary linear convection-diffusion-equation with
     source term", Int. J. Comput. Math. $(2020)$, $24$ pages. DOI: 10.1080/00207160.2020.1726896.

     \bibitem{en11}
    E. Ngondiep. "A fourth-order two-level factored implicit scheme for solving two-dimensional unsteady transport equation with time dependent
     dispersion coefficients", submitted.

    \bibitem{en7}
    E. Ngondiep. "A three-level explicit time-split MacCormack method for $2$D nonlinear reaction-diffusion equations", preprint
    available from http://arxiv.org/abs/$1903.10877$ $(2019)$ $25$ pages.

    \bibitem{en14}
     E. Ngondiep. "A novel three-level time-split MacCormack method for solving two-dimensional viscous coupled Burgers equations", preprint
     available online from http://arxiv.org/abs/1906.01544, 2019.

     \bibitem{nkma}
    E. Ngondiep, N. Kerdid, M. Abdulaziz Mohammed Abaoud, I. Abdulaziz Ibrahim Aldayel. "A three-level time-split MacCormack method for two-dimensional
     nonlinear reaction-diffusion equations", Int J Numer Meth Fluids, $(2020)$ p. $1$-$26.$ https://doi.org/10.1002/fld.4844.

    \bibitem{30tls}
    E. Ngondiep, R. Alqahtani, J. C. Ntonga. "Stability analysis and convergence rate of MacCormack scheme for complete shallow water
    equations with source terms". Preprint available online from http://arxiv.org/abs/1903.11104, 2019.

     \bibitem{tg17}
    B. J. Noye, H. H. Tan. "Finite difference methods for solving the two-dimensional advection-diffusion equation", Int. J. Numer. Methods Fluids,
    $9(1)$ $(1989)$ $75$-$98.$

     \bibitem{31ifvr}.
     World Organization Health. "Coronavirus disease $(2019)$ situation reports".
     https://www.who.int/emergencies/diseases/novel-coronavirus-2019/situation-reports/, March $2020$.

      \bibitem{35ifvr}.
     World Organization Health. "Report of the WHO-China joint mission on Coronavirus disease $2019$".
     https://www.who.int/docs/default-source/coronaviruse/who-china-joint-mission-on-covid-19-final-report.pdf/, February $2020$.

     \bibitem{10cp}
     S. M. Pautanen, D. E. Low, B. Henry, et al. "Identification of severe acute respiratory syndrome in Canada", N. Engl J. Med.,
     April $(2003)$. https://www.nejm.org/

     \bibitem{43ifvr}
    H. R. Thieme. "Mathematics in population biology. Mathematical biology series", Princeton University press, $(2003)$.

     \bibitem{44ifvr}.
     P. Van den Driessche, J. Watmough. "Reproduction numbers and sub-threshold epidemic equilibria for compartmental models of disease transmission",
     Mathematical Bioscience, $180(1-2)$, $(2002),$ p. $29$-$48$.

    \bibitem{45ifvr}.
     R. Verity, L. C. Okell, I. Dorigatti, P. Winskill, C. Whattaker, N. Imai, G. Guomo-Dannenburg, H. Thompson, P. Wlker and et al.. "Estimates of
     the severity of Covid-19 disease.", $medRxiv,$ $(2020)$.

    \bibitem{48ifvr}
    W. Wang, M. D. Aili Jiang, Q. Qin. "Temperature significantly change COVID-19 transmission in $429$ cities", $MedRxiv$, $(2020)$.

    \bibitem{zl1984}
    Z. Zlatev, R. Berkowicz, L. P. Prahm. "Implementation of a variable stepsize variable formula in the time-integration part of a code for
     treatment of long-range transport of air polluants". J. Comput. Phys., $55$ $(1984)$ $278$-$301.$

     \end{thebibliography}
     \end{document}